\documentclass[acmsmall, natbib=false, nonacm=true]{acmart}

\acmJournal{TOCL}

\usepackage{enumitem}
\setlist[description]{
  font=\sffamily\bfseries\itshape
}

\usepackage{amsthm}
\usepackage{thmtools}

\declaretheoremstyle[
  spaceabove=6pt, spacebelow=6pt,
  bodyfont = \normalfont
]{plain}
\declaretheoremstyle[
  numbered=no,
  headfont = \normalfont\itshape,
  bodyfont = \normalfont,
  name=Proof,
  qed=$\triangle$
]{claimproofstyle}

\declaretheorem[style=plain,name=Theorem]{theorem}
\declaretheorem[numberlike=theorem, name=Proposition]{prop}
\declaretheorem[numberlike=theorem]{lemma}
\declaretheorem[numberlike=theorem, name=Corollary]{cor}

\declaretheorem[numberlike=theorem]{definition}
\declaretheorem[numberlike=theorem]{claim}
\declaretheorem[numberlike=theorem, qed=$\boxtimes$]{example}
\declaretheorem[style=claimproofstyle]{claimproof}

\usepackage{tikz}

\newcommand{\negline}{\vspace*{-\baselineskip}}

\newcommand{\tuple}[1]{{\bar#1}}

\newcommand{\cM}{\mathcal{M}}

\newcommand{\cP}{\mathcal{P}}

\newcommand{\cS}{\mathcal{S}}
\newcommand{\cT}{\mathcal{T}}

\newcommand{\cV}{\mathcal{V}}

\newcommand{\fA}{\mathfrak{A}}

\newcommand{\fD}{\mathfrak{D}}

\newcommand{\fN}{\mathfrak{N}}

\newcommand{\talpha}{\tuple{\alpha}}
\newcommand{\tbeta}{\tuple{\beta}}

\newcommand{\tmu}{\tuple{\mu}}

\newcommand{\ta}{\tuple{a}}
\newcommand{\tb}{\tuple{b}}
\newcommand{\tc}{\tuple{c}}

\newcommand{\tu}{\tuple{u}}
\newcommand{\tv}{\tuple{v}}
\newcommand{\tw}{\tuple{w}}
\newcommand{\tx}{\tuple{x}}
\newcommand{\ty}{\tuple{y}}
\newcommand{\tz}{\tuple{z}}

\newcommand{\fo}{\mathrm{FO}}
\newcommand{\fot}{{\fo^{\mathsf{T}}}}
\newcommand{\fom}{{\fo^{\mathsf{M}}}}
\newcommand{\fomts}{{\fo^{\mathsf{mts}}}}
\newcommand{\esomts}{{\mathrm{ESO}^{\mathsf{mts}}}}
\newcommand{\eso}{\operatorname{\Sigma_1^1}}

\newcommand{\dep}{\operatorname{dep}}
\newcommand{\const}{\operatorname{const}}
\newcommand{\excl}{\mathbin{\operatorname{|}}}
\newcommand{\incl}{\subseteq}
\newcommand{\equiex}{\Join}
\newcommand{\indep}{\perp}
\newcommand{\anon}{\Upsilon}
\newcommand{\mindep}{\mathbin{\perp\!\!\!\perp}}
\newcommand{\mincl}{\subsetpeq}

\newcommand{\fork}[1]{\sphericalangle_{#1}}
\newcommand{\cycle}[2]{\operatorname{cycle}(#1, #2)}

\newcommand{\A}{\forall}
\newcommand{\E}{\exists}
\newcommand{\Land}{\bigwedge}
\newcommand{\Lor}{\bigvee}
\newcommand{\imp}{\rightarrow}
\newcommand{\lra}{\leftrightarrow}
\newcommand{\true}{\mathsf{true}}
\newcommand{\false}{\mathsf{false}}
\renewcommand{\emptyset}{\varnothing}
\renewcommand{\phi}{\varphi}
\renewcommand{\theta}{\vartheta}
\renewcommand{\epsilon}{\varepsilon}
\renewcommand{\models}{\mathbin{\vDash}}
\newcommand{\nmodels}{\mathbin{\nvDash}}
\newcommand{\defiff}{\mathbin{:\!\!\iff}}
\newcommand{\ceq}{\mathbin{:=}}
\newcommand{\cceq}{\mathbin{::=}}
\newcommand{\res}[1]{{\upharpoonright}_{#1}}

\newcommand{\error}[1]{\mathsf{e}_{#1}}
\newcommand{\bound}{\mathbin{\preceq}}

\newcommand{\N}{\mathbb{N}}

\newcommand{\R}{\mathbb{R}}

\newcommand{\tsv}{\textsf{tsv}}
\newcommand{\mset}[1]{\{\mkern-5mu\{#1\}\mkern-5mu\}}
\DeclareMathOperator{\dom}{dom}

\DeclareMathOperator{\pr}{Pr}

\newcommand{\ppot}[1]{\cP^+(#1)}
\newcommand{\supp}[1]{T(#1)}
\newcommand{\teamVer}[1]{#1^{\mathsf{T}}}

\DeclareMathOperator{\free}{free}

\newcommand{\quotes}[1]{``#1''}
\newcommand{\cupdot}{
\begin{tikzpicture}[baseline=-.65ex]
  \node (a) at (0, 0) {${\cup}$};
  \node (b) at (0, 0) {${\cdot}$};
\end{tikzpicture}}
\newcommand{\subsetpeq}{
  \mathbin{
  \begin{tikzpicture}[baseline=-.65ex]
    \node[inner sep = 0] (a) at (0, 0) {${\subseteq}$};
    \node[inner sep = 0] (b) at (-.2ex, .125ex) {\fontsize{1.3ex}{0} ${+}$};
    \useasboundingbox (a.north west) rectangle (a.south east);
  \end{tikzpicture}
  }
}

\newcommand{\Iff}{\Longleftrightarrow}
\newcommand{\ra}{\rightarrow}
\newcommand{\NPtime}{\textsc{\textsf{NP}}}

\newcommand{\teamG}[3]{G_{#1}^{#2,#3}}

\begin{document}


\title{Logics with Multiteam Semantics}

\author{Erich Gr\"adel}
\orcid{0000-0002-8950-9991}
\affiliation{%
  \institution{Mathematical Foundations of Computer Science, RWTH Aachen University}
  \city{Aachen}
  \country{Germany}
}
\email{graedel@logic.rwth-aachen.de}

\author{Richard Wilke}
\authornotemark[0]
\authornote{Supported by the DFG RTG 2236 UnRAVeL}
\email{wilke@logic.rwth-aachen.de}
\orcid{0000-0002-8110-0921}
\affiliation{%
  \institution{Mathematical Foundations of Computer Science, RWTH Aachen University}
  \city{Aachen}
  \country{Germany}
}



\begin{abstract}
Team semantics is the mathematical basis of modern logics of dependence and independence.
In contrast to classical Tarski semantics, a formula is evaluated not for a single assignment
of values to the free variables, but on a set of such assignments,  called a team.
Team semantics is appropriate for a purely logical understanding of dependency notions,
where only the presence or absence of data matters, but based on sets, it does not take into account
multiple occurrences of data values. It is therefore insufficient in scenarios where such multiplicities
matter, in particular for reasoning about probabilities and statistical independencies.
Therefore, an extension from teams to multiteams (i.e.~multisets of assignments) has been
proposed by several authors.

In this paper we aim at a systematic development of logics of dependence and independence based on multiteam semantics.
We study atomic dependency properties of finite multiteams and discuss the appropriate meaning of logical operators to extend the atomic dependencies to full-fledged logics for reasoning about dependence properties in a multiteam setting.
We explore properties and expressive power of a wide spectrum of different multiteam logics and compare them to second-order logic and to logics with team semantics.
In many cases the results resemble to what is known in team semantics,
but there are also interesting differences.
While in team semantics, the combination of inclusion and exclusion dependencies leads to
a logic with the full power of both independence logic and existential second-order logic,
independence properties of multiteams are not definable by any combination of properties
that are downwards closed or union closed and thus strictly more powerful
than inclusion-exclusion logic.
We also study the relationship of logics with multiteam semantics with existential second-order
logic for a specific class of metafinite structures. It turns out that inclusion-exclusion logic
can be characterised in a precise sense by the Presburger fragment of this logic, but
for capturing independence, we need to go beyond it and add some form of multiplication.
Finally we also consider multiteams with weights in the reals  and study
the expressive power of formulae by means of  topological properties.
\end{abstract}

\maketitle
\clearpage

\section{Introduction}

Team semantics, originally invented by Hodges \cite{Hodges97} to provide a compositional, model-theoretic semantics for
independence-friendly logic (IF) \cite{MannSanSev11},  has become the mathematical basis of a wide variety of logics for reasoning 
about dependence, independence, or imperfect information. In modern logics of this kind,
dependencies between variables are not presented anymore by annotations of quantifiers (as in IF logic)
but are considered, following a proposal by V\"a\"an\"anen \cite{Vaananen07}, as atomic properties of teams.
A team is a set of assignments $s\colon\{x_1,\dots,x_k\}\to A$
with the same domain of variables and the same co-domain of values, typically the universe of a structure.
Indeed, statements of the form ``$y$ depends on $x$'' or ``$x$ and $y$ are independent'' do not make sense
for an individual assignment of values to $x$ and $y$, but require a larger amount of data, which can be 
modelled by a table or relation, or equivalently by a team of assignments. 

Team semantics has turned out to be extremely fruitful and important for reasoning about dependence and independence.
It has lead to a genuinely new area in logic,
with an interdisciplinary motivation of providing logical systems for reasoning about 
the fundamental notions of dependence and independence that permeate many scientific disciplines. 
Methods from several areas of computer science, including finite model theory, database theory, and the algorithmic analysis of
games have turned out as highly relevant for this area. For more information, we refer to the 
volume \cite{DagstuhlDependenceBook} and the references therein.

However, one of the limitations of team semantics is that, based on sets,
it does not take into account multiplicities of assignments.
Team semantics  thus provides a purely logical understanding of dependency notions,
where only the presence or absence of data matters.
For instance, the independence of two variables $x$ and $y$ in a team $X$ means that
the values for $x$ and $y$ occur in all conceivable combinations in $X$ or, equivalently,
that learning anything new about the actual value of one variable, does not provide any
additional information about the possible values for the other \cite{GraedelVaa13}.
However, it may very well be that $x$ and $y$ are independent in this sense, but not in a statistical sense,
so that additional knowledge about $x$ may lead to new insights into the \emph{probabilities} for the possible values of $y$.
Team semantics is therefore insufficient in scenarios where multiplicities matter,
in particular for reasoning about probabilities and statistical independencies.
A natural semantical framework for logics of dependence and independence taking into account
multiplicities of data is obtained by considering \emph{multisets}, rather than sets, of assignments.
This is also motivated by the prominent role of \emph{bag semantics} in databases,
which is one of the areas in computer science where reasoning about dependencies
is an important issue \cite{VardiCha93, JayramKolVee06}.
Indeed, the standard way of evaluating database queries in commercial systems
(like \textsf{SQL}) is based on bag semantics, and does not eliminate duplicates
in the query result, unless this is explicitly required by the user.
There are several reasons for this, one is that eliminating duplicates may be
computationally expensive and another one is that the multiplicities actually
affect the semantics, as for example when one queries which feature is the most
common among the present data.

\subsection*{Objectives of this paper}
We aim at a systematic development of logics of dependence and independence based on
multiteam semantics. We study atomic dependency properties of finite multiteams and 
discuss the appropriate meaning of logical operators for multiteam semantics, so as to extend
the atomic dependencies to full-fledged logics for reasoning about dependence and
independence in a multiteam setting. 
We compare the properties and
expressive power of a number of different logics with team and multiteam semantics,
and show that in some important aspects, multiteam semantics is quite different from team semantics.
For instance, while in team semantics, the combination of inclusion and exclusion dependencies leads to
a logic with the full power of both independence logic and existential second-order logic,
this is not the case for logics with multiteam semantics. In this
setting independence is not definable by any combination of properties
that are downwards closed or union closed and thus strictly more powerful
than inclusion-exclusion logic. 
We shall also study the relationship of logics with multiteam semantics with 
metafinite model theory, more precisely, with existential second-order
logic for a specific class of metafinite structures. It turns out that inclusion-exclusion logic
can be characterised in a precise sense by the Presburger fragment of this logic, but 
for capturing independence, one would need to go beyond it and  add some form of multiplication.
Finally, we shall study multiteams with weights in the real numbers, which leads to the notions of topologically 
open and closed formulae. Topological properties can be used
to separate certain variants of multiteam logics with respect to expressive power.
Interestingly, this also provides additional insights into the connections between 
team semantics and multiteam semantics with natural multiplicities.

\subsection*{Related work}
The need and the potential of generalising team semantics by taking into account multiplicities
have been noted before, and several approaches for specific applications have been proposed.
Hyttinen et al.\ \cite{HyttinenPaoVaa15} have introduced quantum team logic, a propositional logic for reasoning about quantum phenomena,
and have provided a semantics based on a specific notion of multiteams, called quantum teams.
In a further paper, the same authors \cite{HyttinenPaoVaa17} have introduce measure teams, and a logic for probabilities
of first-order properties of these.

A closely related approach to multiteam semantics is the study of \emph{probabilistic teams}
where the individual assignments come with a probability rather than a multiplicity.
Durand et al.\ \cite{DurandHanKonMeiVir18a} study probabilistic team semantics with
rational probabilities and compare the power of different logics in this setting.
Probabilistic teams have been further studied in \cite{HannulaHirKonKulVir19}, with probabilities that are not
confined to the rational numbers. Both papers also make connections to variants of existential 
second-order logic, that are somewhat similar to the ones established here, although our setting and our
methods are different. Probabilistic team semantics is an important research direction, 
which clearly has similarities to multiteam semantics, but also relevant differences.
Probabilistic notions for teams abstract away from the concrete multiplicities of data and 
instead focus exclusively on relative probabilities, whereas in our work  
the actual multiplicities are taken into account, which in many application scenarios, such as
databases, is of importance.
We believe that both approaches, probabilistic team semantics and multiteam semantics,  
are worth investigating in detail and we here contribute to the latter case.

Multiteams as such have been considered by Durand et at.\ in \cite{DurandHanKonMeiVir18},
but in a different setting where not only the data (team) is equipped with multiplicities, but also the elements of
the underlying structure, so that also the universe is a multiset, rather than a set.
Here instead, we have a purely logical
view on the structure and attach multiplicities to the data records only, i.e. to the
assignments in a multiteam. The differences of the two approaches also lead to different
definitions for the meaning of logical operators.
In \cite{DurandHanKonMeiVir18} the authors study a property of multiteam formulae 
that they call \emph{weak flatness}
which basically means that a formula holds under a multiteam if, and only if, it holds
under a multiteam based on the same team, but where all assignments occur just once,
and they show that weak flatness is a property shared by formulae
in which only downwards closed dependencies  occur.
A multiteam in which every assignment occurs at most once may be considered as a
team, hence this result builds a bridge between team and multiteam semantics.
In the present paper we also consider links between these two formalisms,
but instead of identifying a fragment of multiteam formulae which in their
meaning coincide with team semantical formulae we study the question which properties
that are expressible in a logic with team semantics can also be expressed
in multiteam semantics (when one has access only to multiteam atoms and not to team
atoms, of course).
A further matter, which is addressed in \cite{DurandHanKonMeiVir18}  and in a slightly different form by 
V\"a\"an\"anen in \cite{Vaananen17}, is the notion
of \emph{approximate operators} $\left<p\right>\theta$ and $[p]\theta$, where
$p$ is a probability.
These formulae are satisfied by a multiteam $M$ if some/all submultiteams $N$ of
$M$ with size at least $p\cdot|M|$ satisfy $\phi$.
Although we do not investigate these operators in this paper, we will introduce
a variety of atomic multiteam formulae, most notably the \emph{forking atoms}
$\tx\fork{\leq p}\ty$ with which it is possible to define the existential
approximate operators in our setting.
Further, Hirvonen et al.\ \cite{HirvonenKonPau19} study multiteams from a different angle, considering open and closed formulae based on 
metric spaces.

\section{Team Semantics}

We recall the standard definitions related to team semantics (see, e.g., \cite{DagstuhlDependenceBook,Galliani12,GraedelVaa13,Vaananen07}).
A \emph{team} $X$ is a set of assignments $s\colon D \to A$ with a common finite domain $D = \dom(X)$
of variables and a common codomain $A$ of values, typically the universe of a structure $\fA$.
For a team $X$ with domain $\dom(X)$ and $V\subseteq \dom(X)$ the \emph{restriction}
of $X$ to $V$ is the team $X\res{V} := \{s\res{V} : s\in X\}$.
We write $s[x\mapsto a]$ for the assignment
that extends, or updates, $s$ by mapping $x$ to $a$.
We write $\mathcal{P}(A)$ for the power set of $A$ and 
set  $\mathcal{P}^+(A) := \mathcal{P}(A) \setminus \{ \emptyset\}$.
Basic operations that possibly extend
the domain of  a given team to new variables are the unrestricted 
\emph{generalisation} over $A$,
$X[x\mapsto A] :=  \{ s[x\mapsto a] : s \in X, a \in A \}$
and the \emph{Skolem-extensions}
$X[x\mapsto F] := \{ s[x\mapsto a] : s \in X, a \in F(s) \}$
for any function $F \colon X \to \mathcal{P}^+(A)$.
Modern logics of dependence and independence are based on atomic properties of teams, that are then extended to full-fledged
logics for reasoning about such dependencies, using the rules of team semantics.

\begin{definition}[Atomic dependencies]
  \label{def: dependency atoms}
  Let $\fA$ be a structure and $X$ a team over $\fA$.
  \begin{description}
    \item[Dependence:] $\fA \models_X \dep(\tx, y)\ \defiff\  (\A s,s'\in X)\ s(\tx) = s'(\tx) \Rightarrow s(y) = s'(y)$; {\hfill\normalfont\cite{Vaananen07}}
    \item[Exclusion:] $\fA \models_X \tx \excl \ty\ \defiff\ (\A s,s'\in X)\ s(\tx) \neq s'(\ty)$; {\hfill\normalfont\cite{Galliani12}}
    \item[Inclusion:] $\fA \models_X \tx\incl\ty\ \defiff\ (\forall s\in X)(\exists s'\in X)\ s(\tx) = s'(\ty)$; {\hfill\normalfont\cite{Galliani12}}
    \item[Equiextension:] $\fA \models_X \tx\equiex\ty\ \defiff\ \fA \models_X \tx \incl \ty$ and $\fA \models_X \ty\incl\tx$; {\hfill\normalfont\cite{Galliani12}}
    \item[Anonymity:] $\fA \models_X \tx\anon y\ \defiff (\A s\in X)(\E s'\in X)\ s(\tx)=s'(\tx)$ and $s(y)\neq s'(y)$; {\hfill\normalfont\cite{DagstuhlDeplog19}}
    \item[Independence:] $\fA \models_X  \tx\indep \ty\ \defiff\ (\A s,s'\in X)(\E s''\in X)\ s''(\tx)=s(\tx), s''(\ty)=s'(\ty)$. {\hfill\normalfont\cite{GraedelVaa13}}
  \end{description}
\end{definition}

An important issue in the study of atomic dependencies are \emph{closure properties}. Dependence
and exclusion atoms are \emph{downwards closed}:  whenever such a dependency is true for a team $X$,
then it is also true for all subteams $Y\subseteq X$. Similarly, inclusion, equiextension, and anonymity atoms are
\emph{closed under unions}:  whenever they are true in all teams $X_i$ (for $i\in I$), then they
also hold in their union $\bigcup_{i\in I} X_i$. Independence atoms instead are neither downwards closed nor
closed under unions.

Given any collection $\Omega$ of atomic team properties, we denote by $\fot[\Omega]$ the closure 
of the atoms from $\Omega$ and the first-order literals (of some fixed relational vocabulary $\tau$)
under the logical operators $\land, \lor, \E$, and $\A$.

The traditional semantics (to which we refer as Tarski semantics) for $\phi\in\fo$
is based on single assignments $s$ whose domain must comprise the variables
in $\free(\phi)$; we write $\fA\models \phi[s]$ for saying that $\fA$ satisfies
$\phi$ with the assignment $s$.
The rules of team semantics extend the semantics of the atomic dependencies to tell us when it is the case that 
$\fA\models_X \phi$ for $\phi(\bar x)\in\fot[\Omega]$, a $\tau$-structure $\fA$ and a team $X$ 
whose domain includes the free variables of $\phi$.

\begin{definition}[Team Semantics]
  \label{def: team semantics}
  \begin{itemize}
    \item If $\phi$ is a literal then $\fA \models_X \phi$ if $\fA \models \phi[s]$ for all $s \in X$;
    \item $\fA \models_X \phi_1 \land \phi_2$ if $\fA \models_X \phi_i$ for $i = 1,2$;
    \item $\fA \models_X \phi_1 \lor \phi_2$ if $X = X_1 \cup X_2$ for two teams $X_i$ such that  $\fA \models_{X_i} \phi_i$;
    \item $\fA \models_X \forall x\phi$ if $\fA \models_{X[x\mapsto A]} \phi$;
    \item $\fA \models_X \exists x\phi$ if $\fA \models_{X[x\mapsto F]} \phi$ for some $F\colon X\to\ppot{A}$.
  \end{itemize}
\end{definition}

Besides this standard variant for team semantics (sometimes called \emph{lax} semantics), there also exists
a \emph{strict} variant, which for disjunctions requires a \emph{disjoint} decomposition $X=X_1\cupdot X_2$,
and for existential quantification a Skolem extension $X[x\mapsto F]$ with $F\colon X\to A$,
which selects for each assignment $s\in X$ a single value for the quantified variable, rather than a non-empty set of values.
The distinction between the two variants is immaterial for logics that are downwards closed (such as 
exclusion or dependence logic) but for inclusion logic or independence logic, the
strict semantics is deficient in the sense that it violates the \emph{locality principle}
that a formula should only depend on the variables occurring in it.
We shall see that this is different for the multiteam semantics developed below;
there a strict variant of the semantics is the `right' one.

One way to understand the expressive power of a logic  with team semantics is to relate it
to some well-understood logic with classical Tarski semantics.
One can translate formulae $\phi(\bar x)$ from a logic $\fot[\Omega]$ with vocabulary $\tau$ 
into \emph{sentences} $\phi^X$, with Tarski semantics, of
vocabulary $\tau \,\cupdot\, \{X\}$ where $X$ is an additional relation symbol for the team, such that
for every structure $\fA$ and every team $X$ we have that
\[  \fA\models_X\phi(\bar x)\ \Longleftrightarrow\ (\fA,X)\models \phi^X,\]
where  $(\fA,X)$ is the expansion of $\fA$ by a relational encoding of the team $X$.
In all logics with team semantics that extend first-order formulae
by atomic dependencies that are themselves first-order definable, 
and which do not make use of additional
connectives beyond $\land,\lor$ and atomic negation
(in particular no connective with which one can define classical negation),
such a translation will always produce sentences
in (a fragment of) existential second-order logic $\Sigma^1_1$. 
Understanding the expressive
power of a logic $\fot[\Omega]$ with team semantics thus means to identify the
fragment of $\Sigma^1_1$ to which $\fot[\Omega]$ is equivalent in the sense just described.
By now, many results of this kind are known; some of the most important ones are the following.
\begin{itemize}
  \item Dependence logic $\fot[\dep]$ and exclusion logic $\fot[{}\excl{}]$ are equivalent to the fragment of  $\Sigma^1_1$-sentences $\psi(X)$ in which the predicate for the team appears only negatively \cite{KontinenVaa09}.
  \item Inclusion logic is equivalent to sentences of form $\forall \bar x (X\bar x\rightarrow \psi(X,\bar x))$ in the posGFP-fragment of least fixed-point logic, such that $X$ occurs only positively in $\psi(X,\tx)$ \cite{GallianiHel13}.
  \item Independence logic $\fot[\indep]$ and inclusion-exclusion logic $\fot[\incl,{}\excl{}]$ are equivalent with full $\Sigma^1_1$ (and thus can describe all NP-properties of teams) \cite{Galliani12}.
  \item $\fot[\cup\text{-game}]$ is equivalent to the fragment of all union closed formulae of $\eso$ \cite{HoelzelWil20}.
\end{itemize}

\section{Atomic Dependencies for Multiteams}
\label{sec: atomic multiteam}

\subsection*{Multisets and multiteams}
A \emph{multiset} (or bag) $M$ is a pair $(S, m)$ where $S$ is a set and $m\colon S\to\N_{>0}$
describes the multiplicities of the elements $s\in S$ in $M$. For $x\notin S$ we define $m(x) \ceq 0$.
Given two function $m\colon S\to\N_{>0}$ and $m'\colon S'\to\N_{>0}$ and a binary operator
$\circ\colon\N_{>0}\times\N_{>0}\mapsto\N_{>0}$ put $m\circ m'\colon S\cup S'\to\N_{>0}, s\mapsto m(s)\circ m'(s)$.
The size of $M$ is $|M| \ceq \sum_{s\in S} m(s)$ and the \emph{additive union}
with $M'=(S',m')$ is $M\uplus M':=(S\cup S',m+m')$.
Their \emph{product} is $M \times N := (S\times S', m\cdot m')$.
Further, $M$ is a \emph{submultiset} of $M'$,
in symbols $M\subsetpeq M'$, if $S\subseteq S'$ and $m(s)\leq m'(s)$ for all $s\in S$.
For $k\in\N$, we set $kM:=(S,k\cdot m)$, where $0M$ is the empty multiteam $(\emptyset, \emptyset)$.
In case we explicitly present a multiset by listing its elements, we use double set braces as in $\mset{a, a, b} = (\{a, b\}, a\mapsto2, b\mapsto1)$.
We also  make use of the notation $\mset{x \in M : \rho(x)}$,
where $M$ is a multiset, to denote the restriction of $M$ to those elements
$x$ that satisfy the condition $\rho$.
Formally, this is the multiset $\biguplus_{x\in M, \rho(x)} m(x)\cdot\mset{x}$.

A \emph{multiteam} is a multiset $M = (X, m)$ where $X$ is a team.
We write $\supp{M}$ for the underlying team $X$ of $M$.
The domain and codomain of $M$ are those of $\supp{M}$ and the restriction of $M$
to the domain $V$ is the multiteam $(\supp{M}\res{V}, m')$, where $m'\colon s\mapsto\sum_{s'\in \supp{M}, s'\res{V} = s}m(s')$.
In this paper we only consider finite multiteams $M$, with $|M|\in\N$, and we
tacitly assume this whenever we talk about multiteams.
The multiset of values of a variable tuple $\tx$ in a multiteam $M$ is
$M(\tx) := \mset{s(\bar x): s\in M}$.
We denote the restriction of $M$ to those assignments that map $\tx$ to $\ta$ by
$M_{\tx = \ta} := \mset{ s\in M : s(\tx) = \ta}$.
For any first-order atom $\beta$ we define
the restriction of a multiteam $M$ to all assignments that satisfy
this formula, by putting $M_\beta := \mset{s\in M : \fA \models \beta[s]}$, where
the structure $\fA$ should be clear in the given context.
For a tuple $\ta$ of elements from $A$, the \emph{probability} that a randomly chosen
assignment of the multiteam $M$ maps $\tx$ to $\ta$ is $\pr_M(\tx = \ta) := | M_{\tx = \ta}|\,/\, |M|$.
Further, we have the conditional probabilities $\pr_M(\tx = \ta \mid \ty = \tb) :=
|M_{\tx\ty = \ta\tb}|\, /\, |M_{\ty=\tb}|$ (which is undefined in case $M_{\ty=\tb}$ is empty).

\begin{definition} Let $M=(X,m)$ be a multiteam with values in $A$. The universal extension of $M$ on a variable $x$ is 
$M[x \mapsto A] \ceq (X[x\mapsto A], m')$ where $m'(s[x\mapsto a]):=m(s)$.
Equivalently, we can write $M[x \mapsto A] =\mset{s[x\mapsto a]: s\in M,a\in A}$.
A choice function $F\colon M\to A$ for $x$ on $M$ assigns to every (occurrence of an)
assignment $s$ in $M$ a value in $A$;
different occurrences  of the same assignment can be mapped to distinct elements.
The Skolem extension $M[x\mapsto F]$ is the multiteam of all assignments $s[x\mapsto a]$
extending (occurrences of ) $s$ in $M$ by the values assigned by $F$.
Formally, $F$ is a function whose domain is the set of all tuples $(s, i)$ satisfying
$s\in\supp{M}$ and $0\leq i< m(s)$, with codomain $A$, which we indicate by writing
$F\colon M\to A$ as an abbreviation.
Accordingly, $M[x\mapsto F]$ is the multiteam $\biguplus_{s\in\supp{M}, i<m(s)}\mset{s[x\mapsto F(s,i)]}$.
Notice that $|M[x\mapsto F]|=|M|$ for every multiteam $M$ and every choice function $F$, whereas for universal extensions, we have 
$|M[x\mapsto A]|=|M|\cdot|A|$.
\end{definition} 

\subsection*{Atomic dependencies}
Of course, we can view any atomic dependency property $\alpha$ for teams also as a property for multiteams
by defining that $\fA \models_M \alpha$ if, and only if, $\fA \models_{\supp{M}} \alpha$.
In fact, for notions such as dependence and exclusion (which are downwards closed statement that do not
refer to multiplicities), this is the natural way to go.
However, for properties such as inclusion, independence or forking, we 
work with different definitions that take multiplicities into account and 
are appropriate for statistical dependencies of data. 
The following atomic properties capture relevant dependencies for multiteams.

\begin{description}
  \item[Multiteam inclusion:]
    $\fA \models_M \tx \mincl \ty$ holds if, and  only if, $M(\tx) \subsetpeq M(\ty)$.
  \item[Restricted inclusion:]
    $\fA \models_M \tx \mincl_\alpha \ty$ if, and only if, $M_\alpha(\tx) \subsetpeq M(\ty)$.
  \item[Forking:]
    For $\vartriangleleft\in\{<, \leq,=,\geq, >\}$ and $p\in[0,1]$ we define
    $\fA \models_M \tx\fork{\vartriangleleft p} \ty$, if, and only if, for all $\ta, \tb \in A^*$ where $\pr_M(\ty = \tb \mid \tx = \ta) > 0$ also $\pr_M(\ty = \tb \mid \tx = \ta) \vartriangleleft p$ holds.
  \item[Statistical independence:]
    $\fA \models_M \tx \mindep \ty$ holds if, and only if, $\pr_M(\tx=\ta) = \pr_M(\tx=\ta\mid\ty = \tb)$ for all $\ta\in M(\tx)$ and 
    $\tb\in M(\ty)$.
    An equivalent condition is $M(\tx)\times M(\ty)=|M| \cdot M(\tx\ty)$.
  \item[Conditional independence:]
    $\fA \models_M \tx \mindep_\tz \ty$ if, and only if, $\fA \models_{M_{\tz=\tc}} \tx\mindep\ty$ holds for all $\tc\in A^{|\tz|}$.
\end{description}

\noindent
Forking atoms have also been considered in the setting of team semantics \cite{GraedelHeg16}.

One can come up with many other atomic dependencies for multiteams.
There is one relevant condition that we shall impose for such dependencies $\alpha$
which is \emph{invariance under multiplication} of data. This means that
for every structure $\fA$ and every multiteam $M$ we have
$\fA\models_M \alpha$ if, and only if, $\fA\models_{kM}\alpha$ for every natural number $k>0$.
Clearly, all dependencies that we have defined satisfy this condition. If it is abandoned, one could define
pathological atoms  (e.g.~one saying that $|M|$ is prime or $|M| < 5$) that would violate some of the
properties of multiteam semantics that we are going to prove in this paper.

\subsection*{Downwards closed atoms}
A multiteam atom $\alpha$ is \emph{downwards closed} if 
$\fA\models_M\alpha$ and $N\subsetpeq M$ imply that also $\fA\models_N\phi$.
Important examples are dependence and exclusion which are taken over from team semantics.
We show next that, in fact, all downwards closed atoms only depend on the underlying team
and are, in this sense, really team semantical atoms.
We shall establish a stronger result in Section \ref{Sect: BandF}.

\begin{prop}
  \label{prop: downward atom = ts atom}
  Let $\alpha$ be a downwards closed atom for multiteams.
  Then for all multiteams $M, N$ with $\supp{M}=\supp{N}$ we have that $\fA\models_M\alpha$ if, and only if, $\fA\models_N \alpha$.
\end{prop}
\begin{proof}
  If $\fA\models_M\alpha$ then, by closure under multiplication, also  $\fA\models_{|N|\cdot M} \alpha$.
  But if $N$ and $M$ are based on the same supporting team, then $N\subsetpeq |N|\cdot M$, so by downwards closure, $\fA\models_N \alpha$.
\end{proof}

\subsection*{Union closed atoms}
Also the notion of union closure naturally extends to multiteams. We call
$\alpha$ \emph{union closed} if $\fA\models_M\alpha$ and $\fA\models_N\alpha$ imply that $\fA\models_{M\uplus N}\alpha$.

\begin{prop}
  \label{prop: mincl basic properties}
  \begin{enumerate}
    \item The atoms $\tx\mincl\ty$, $\tx \mincl_\alpha \ty$ and  $\tx\fork{\leq p} \ty$ are union closed.
    \item $\tx\mincl\ty \equiv \ty\mincl\tx$.
    \item $\fA\models_M \tv\tx\mincl\tv\ty$ if, and only if, $\fA\models_{M_{\tv=\ta}}\tx\mincl\ty$ for all $\ta\in M(\tv)$.
    \item If $\fA\models_M \tx\mincl\ty$ and $\fA\models_{M\uplus N} \tx\mincl\ty$, then also $\fA\models_N\tx\mincl\ty$.
    \item\label{prop: mincl basic properties - decompose}
    $\fA\models_M \tx\mincl\ty$ if, and only if, $M$ can be decomposed into (finitely many) multiteams $M_i$ such that $M = \biguplus_iM_i$, every multiplicity occurring in an $M_i$ is 1 and $\fA\models_{M_i} \tx\mincl\ty$
  \end{enumerate}
\end{prop}

\begin{proof}
Claim (1) is easily checked.  For Claim (2) we notice that for any tuples $\tx$ and $\ty$ of the same length, we have
$|M(\tx)| = |M(\ty)|$ in every multiteam $M$.
If $\tx\mincl\ty$ holds in $M$ then, for every value $\ta$ we have that $|M_{\tx=\ta}| \leq | M_{\ty=\ta}$|.
In fact we then have that $|M_{\tx=\ta}| = | M_{\ty=\ta}|$, because otherwise there must exist some other value
$\tb$ with $|M_{\tx=\tb}| > | M_{\ty=\tb}|$, contradicting the assumption that $\tx\mincl\ty $ holds in $M$.
(This has also been noticed in \cite{DurandHanKonMeiVir18}).
The third claim follows immediately from the definition.
 Towards Claim (4),  assume that $\fA\models_M\tx\mincl\ty$ and $\fA\models_{M\uplus N}\tx\mincl\ty$.
  Let $\ta$ be any value occurring for $\tx$ in $N$.
  Since $\fA\models_M\tx\mincl\ty$ we have, by (2), that $|M_{\tx = \ta}| = |M_{\ty = \ta}|$.
  But since also $|(M\uplus N)_{\tx = \ta}| = |(M\uplus N)_{\ty = \ta}|$, and
  obviously $|(M\uplus N)_{\tx=\ta}| = |M_{\tx=\ta}| + |N_{\tx=\ta}|$, it follows
  that also  $|N_{\tx = \ta}| = |N_{\ty = \ta}|$.
  For (5), let $\fA\models_{(X, n)}\tx\mincl\ty$ and pick $s\in X$ such that
  $n(s) > 1$ and $s(\tx) \neq s(\ty)$ (if no such $s$ exists, then the claim is obviously true).
  Since $(X, n)$ satisfies $\tx\mincl\ty$ there must be an assignment $s'\in X$
  for which $s(\ty) = s'(\tx)$ holds.
  By repeating this argument we find a chain $s = s_0, s_1, \dotsc$ of assignments
  in $X$ satisfying $s_{i-1}(\ty) = s_i(\tx)$ for all $i>0$.
  Because $(X, n)$ is finite there must be a smallest $k>0$ such that $s_k(\ty)$
  has already occurred as $s_j(\tx)$ for some $j<k$.
  This means that $\hat{M} = \mset{s_j, s_{j+1}, \dotsc, s_k}$ satisfies $\tx\mincl\ty$
  and all assignments in this multiset are pairwise distinct.
  Applying (4) with $M = \hat{M}$ and $M \uplus N = (X, n)$ allows us to repeat
  this argument recursively which proves Claim (5).
\end{proof}

\section{Multiteam Semantics for First-Order Logic}

In the same way as team semantics generalises from a set $\Omega$ of atomic dependencies for teams
to logics $\fot[\Omega]$ we would also like to have full-fledged logics for reasoning about
dependency properties of multiteams. Syntactically, a collection $\Omega$ of atomic
dependencies of \emph{multiteams} is extended to logics $\fom[\Omega]$ in a completely
analogous way, by closing the atoms from $\Omega$ and the first-order literals (of some vocabulary $\tau$),
by $\lor, \land, \E$ and $\A$, thus insisting that formulae are in negation normal form, and that
negation is applied only to first-order atoms, not to atomic dependencies.

The multiteam semantics of first-order literals is as in the case of team semantics, i.e.~for every structure $\fA$, every multiteam $M$, and every 
first-order literal $\alpha(\bar x)$ we define that $\fA\models_M \alpha$ if, and only if, $\fA\models\alpha[s]$ for all $s\in M$.
However, to come up with appropriate definitions for the multiteam semantics of the logical operators
is not completely trivial and
we have to justify the choices that we make. To do this we shall formulate certain \emph{postulates}
that an appropriate multiteam semantics should satisfy, and we shall argue 
that our choices are essentially enforced by these postulates. Of course, the postulates themselves cannot be `proved';
however, we hope to convince the readers that they reflect natural properties that should be satisfied by
reasonable semantical rules, and that violating them leads to pathological or undesired effects.

\subsection{Postulates and definitions for multiteam semantics}

The first postulate is the \emph{locality principle}, familiar from team semantics,  
that the meaning of a formula should depend only on those variables that occur free in it.

The second postulate states that multiteam semantics should be, in a reasonable sense, compatible with team semantics. 
Since every atom for teams can also be used as an atom for multiteams, we have for every collection $\Omega$
of team semantical atoms, a natural embedding of $\fot[\Omega]$ in $\fom[\Omega]$.
We postulate that every $\phi\in\fot[\Omega]$ that holds for a team $X$ also holds
for \emph{some} multiteam with base set $X$ and vice versa.
The reasoning behind this postulate is that multiteam semantics provides a more
detailed view on the data than team semantics, hence, if a statement holds under
multiteam semantics it should \emph{a fortiori} hold under plain team semantics.
On the other hand, knowing that a purely team semantical property is satisfied by
$X$ it should be possible to provide adequate multiplicities to the assignments
in $X$ such that the same property holds under multiteam semantics.

The third postulate concerns disjunctions; even in team semantics, disjunction is in general not idempotent, and a well-known example
witnessing this is the constancy atom $\const(x):=\dep(\emptyset,x)$. Indeed, while $\const(x)$ is true in a team $X$ if, and only if,
the value of $x$ is constant in $X$, the formula $\const(x)\lor\const(x)$ holds in $X$ if, and only if, $x$ takes at most two values 
in $X$. A necessary and sufficient condition for $\phi \equiv \phi \lor \phi$ is that
$\phi$ is \emph{closed under (finite) unions}.  We postulate that this equivalence should also hold in multiteam semantics.

Our final postulate concerns a natural property of quantifiers: they should not
affect atomic dependencies that do not involve the quantified variable.

Here is the list of our postulates:
\begin{description}
  \item[P1: Locality.]\label{itm:locality}\quad
  $\fA\models_M \phi$ if, and only if, $\fA\models_{M\res{\free{\phi}}} \phi$.
  \item[P2: Compatibility with team semantics.]\quad For any team-semantical formula $\phi\in\fot[\Omega]$ we have that
  $\fA \models_X \phi$ if, and only if, $\fA \models_{(X,n)} \phi$ for some $n\colon X\to\N_{>0}$.
  \item[P3: Idempotence of disjunction.]\quad
  $\phi \equiv \phi \lor \phi$ if, and only if, $\phi$ is \emph{union closed}.
  \item[P4: Dummy quantification preserves atomic dependencies.]\quad
  For both $Q=\exists$ and $Q=\forall$ we have that $ Q x\,\alpha \equiv \alpha$ for
  all multiteam atoms $\alpha$ in which the variable $x$ does not occur.
\end{description}

Based on these postulates we define multiteam semantics by the following rules.

\begin{definition}
  \label{def: multiteam semantics} 
  Let $\Omega$ be a collection of multiteam dependency atoms, $\fA$ a structure, $M$ a multiteam over $A$ and 
  $\phi,\psi\in\fom[\Omega]$. 
  \begin{itemize}
    \item $\fA \models_M \alpha$ if $\fA\models\alpha[s]$ for all $s\in M$ if $\alpha$ is a first-order literal;
    \item $\fA \models_M \phi \land \psi$ if $\fA \models_M \phi$ and $\fA \models_M \psi$;
    \item $\fA \models_M \phi \lor \psi$ if there are $M_1\uplus M_2=M$ with $\fA \models_{M_1} \phi$ and $\fA\models_{M_2} \psi$;
    \item $\fA \models_M \forall x\psi$ if $\fA \models_{M[x\mapsto A]} \psi$;
    \item $\fA \models_M \exists x\psi$ if $\fA \models_{M[x\mapsto F]} \psi$ for some choice function $F\colon M\to A$.
  \end{itemize}
\end{definition}

By induction one can show that first-order formulae under multiteam
semantics are \emph{flat}, that means $\fA\models_M\phi$ if, and only if,
$\fA\models_{\mset{s}}\phi$ for all $s\in\supp{M}$.
Further $\fA\models_{\mset{s}}\phi$ holds precisely if $\fA\models\phi[s]$
under Tarski semantics, hence we generalise the notation $M_\phi$ to arbitrary
first-order formulae, not just to atoms, with the meaning $M_\phi=\mset{s\in M : \fA\models\phi[s]}$.

Classical negation is not part of our syntax.
But since negation is allowed in front of classical first-order atoms and the
dual of every operator is present, first-order formulae without dependencies can
always be transformed into negation normal form.
For first-order $\zeta$, we let $\zeta \ra \chi$
be a shorthand for $\operatorname{nnf}(\neg\zeta)\lor (\zeta\land\chi)$.
Note that $\zeta$ occurs on both sides of the disjunction to ensure that precisely
those assignments are evaluated against $\chi$ that satisfy $\zeta$.
One can check that this implication satisfies $\fA \models_M \zeta \ra \chi$
if, and only if, $\fA \models_{M_\zeta}\chi$.
It would mean something different if we would have defined it to be equivalent
to $\operatorname{nnf}(\neg\zeta) \lor \chi$, since then we would only have that
the implication is satisfied by $\fA$ and $M$ in case there is a multiteam $N$
which is sandwiched between $M_\zeta$ and $M$, i.e.~$M_\zeta\subsetpeq N\subsetpeq M$,
such that $\fA \models_N \chi$, but this is not what we want.
The following example illustrates an application of this implication.

\begin{example}
  For a finite directed graph $G = (V, E)$, we have that $G\models_{\mset{\emptyset}}  \A v\A w(Evw \ra v\mincl w)$ if every vertex in $G$ has the same number of incoming as outgoing edges.
  
  Indeed, after universally quantifying the variables $v$ and $w$ one ends up with
  the multiteam $M = \mset{\emptyset}[v\mapsto V][w\mapsto V]$ which holds for every
  pair of vertices $x,y\in V$ an assignments $s_{x,y}\colon v\mapsto x, w\mapsto y$.
  Thus, especially $M$ has size $|V|^2$.
  Furthermore, by the arguments above we have $G \models_M Evw \ra v\mincl w$ if,
  and only if, $G\models_{M_{Evw}} v\mincl w$.
  Now this just means that we have to count every edge in $G$ and have to check
  whether every vertex appears equally often as an ingoing as an outgoing
  vertex.
\end{example}

In team semantics, formulae expressing properties that hold only for a certain
portion of the given team have been investigated in \cite{Vaananen17} (e.g.\ $X
\models \dep_{1/2}(x, y)$ if half of $X$ satisfies the dependence $\dep(x, y)$)
and similarly also in context of multiteam semantics in \cite{DurandHanKonMeiVir18},
where an \emph{approximate operator} $\left<p\right>$ has been introduced with
the semantics $\fA\models_M\left<p\right>\phi$ if there is a multiteam
$N \subsetpeq M$ with $p|M| \leq |N|$ such that $\fA\models_N \phi$.

\begin{example}
  In our context, for any $\psi$, let $\psi_{1/2} \ceq \E z(\fork{=1/2}z \land
  ((\dep(z) \land \psi) \lor \dep(z)))$, which holds in $M$ under $\fA$ if
  $\fA\models_N\psi$ for some $N \subsetpeq M$ with $|N| = |M|\,/\,2$.
  Now, this formula can easily be adapted for other probabilities.
\end{example}

\subsection{The complexity of evaluation in multiteam semantics}

It is often useful to describe the evaluation process for determining
whether  $\fA\models_M\psi$ by an \emph{annotation}  of the syntax tree $\cT(\psi)$
of $\psi$, which associates with every subformula $\phi\in\cT(\psi)$ a multiteam $M(\phi)$.
We call such an annotation \emph{structurally valid} (for $\fA, M$, and $\psi$) if $M(\psi)=M$ and
the annotation respects the semantic rules of Definition \ref{def: multiteam semantics} as follows:
\begin{itemize}
\item $M(\phi\land\theta)=M(\phi) =M(\vartheta)$; $M(\phi\lor\theta)=M(\phi)\uplus M(\vartheta)$;
\item for $\phi=\A x\theta$, we have that $M(\theta)= M(\phi)[x\mapsto A]$, and
\item for $\phi=\E x\theta$, we have that $M(\theta)= M(\phi)[x\mapsto F]$ for some choice function $F\colon M\to A$.
\end{itemize}

\begin{prop}
  \label{prop:annotation}
  For all $\fA, M$, and $\psi$, we have that $\fA\models_M \psi$ if, and only if, there exists a structurally valid annotation of the syntax tree $\cT(\psi)$, such that, for every leaf $\phi$ of $\cT(\psi)$ (which is either a first-order literal, or an atomic dependency) we have that $\fA\models_{M(\phi)} \phi$.
\end{prop}

For complexity analysis, we encode multiteams $M=(X,m)$ via binary representations of the numbers $m(s)$, for $s\in X$.
We observe that the data complexity of $\fom$ is in \NPtime, since for any fixed formula, annotations of polynomial size can be guessed and efficiently verified.
Further, it is well-known that even $\fot[\dep]$ can express  \NPtime-complete problems \cite{Vaananen07}.

\begin{cor}
  \label{cor:NP}
  Let $\Omega$ be any collection of multiteam atoms that can be evaluated in polynomial time on any given $\fA$ and $M$.
 Then, for every $\psi\in\fom[\Omega]$, the problem whether $\fA\models_M \psi$ for given $\fA$ and $M$, can be solved in \NPtime. Further, whenever $\Omega$ contains $\dep$, $\mid$, or $\mindep$, then
there are formulae in $\fom[\Omega]$, for which the evaluation problem is \NPtime-complete.
\end{cor}
\begin{proof}
  By Proposition~\ref{prop:annotation}, we can determine whether $\fA\models_M \psi$ by guessing an annotation
  of the syntax tree $\cT(\psi)$ and verifying that it is structurally valid and satisfies the atoms. Since $\psi$ is fixed,
  so is the syntax tree. Going from a node of $\cT(\psi)$ to its successors can only decrease the multiplicities in the multiteams.
  In particular, the multiplicity of an assignment $s[x\mapsto a]$ in $M(\phi)$ is not larger than the multiplicity
  of $s$ in $M(\E x\phi)$ or $M(\A x\phi)$.
  Since the number of variables in subformulae of $\psi$ is bounded by a constant,
  an annotation of $\cT(\psi)$ can be represented with polynomially bounded length, with respect to the
  length of the inputs $\fA$ and $M$, and clearly the verification that 
  an annotation is structurally valid and satisfies the atoms can be done in polynomial time. This proves that the evaluation problem
  for any fixed formula $\psi\in\fom[\Omega]$ is in \NPtime.
  Finally since there even are team semantical formulae describing \NPtime-complete problems, this holds \emph{a fortiori}
  for  $\fom[\Omega]$, when $\Omega$ contains, for instance dependency atoms $\dep(x,y)$.
\end{proof}

\subsection{Justification of the semantical rules}

For the semantics of first-order atoms, conjunction and universal quantification
there are no reasonable alternatives to the semantics defined above,
in particular when keeping in mind our second postulate about compatibility with team semantics. 
However, for the existential operators, i.e.~disjunction and existential quantification,
there are other options that one has to consider when defining the semantics.
We  elaborate on these and justify the choices we made in the following paragraphs.

\subsection*{Disjunction}
As in team semantics, the disjunction must (again due to postulate two) in some
way allow to split the multiteam into two parts each of which satisfies one of
the disjuncts.
There are two possibilities: the first one
just requires that the two parts cover the given multiteam, 
but permits to put an individual assignment into both parts.
\begin{equation}
  \label{eq:cover disjunction}
  \fA \models_M \phi \lor \psi\quad\Iff\quad
  \fA \models_S \phi \text{~and~} \fA \models_T \psi
  \quad\text{for some~} S,T \subsetpeq M \subsetpeq S\uplus T
  \tag{$\lor$-cover}
\end{equation}
The other possibility is to require a disjoint split of the
multiteam into two parts.
\begin{equation}
  \label{eq:split disjunction}
  \fA \models_M \phi \lor \psi \quad\Iff\quad
  \fA \models_S \phi \text{~and~} \fA \models_T \psi
  \quad\text{for some~} S,T \subsetpeq M = S\uplus T
  \tag{$\lor$-split}
\end{equation}

The choice between the two possibilities
is irrelevant for downwards closed formulae since for these the two variants
\eqref{eq:cover disjunction} and \eqref{eq:split disjunction}
lead to equivalent meanings for the disjunction. However, for formulae that are not downwards closed,
such as inclusion and independence statements, it does make a difference.
Recall that in team semantics there is a similar distinction between
the strict and lax semantics for disjunction \cite{Galliani12}. There it turns out that
the lax semantics, allowing splits that are not disjoint, is the right
choice, because otherwise the locality postulate is violated for certain formulae.
We argue that for multiteam semantics the situation
is different, and a disjoint split is the good choice.
On an intuitive level, if data is modelled by teams then
information about multiplicities of data is omitted, and it may be necessary to
assume that the data must be present in both parts of a split.
But in multiteams, information about multiplicities is available,
and if we recombine two parts of a multiteam then we have the sum of their multiplicities,
so it is much more natural to allow each assignment to be contained in precisely 
one of the parts of a splits.
On a more formal level we justify the choice for ($\lor$-split)
by the postulate of equivalence between union closure and idempotent disjunction.

\begin{lemma}
  \label{lem: mts union closed iff idempotent disjunction}
  Let $\phi$ be a formula with multiteam semantics. Modelling disjunctions via  ($\lor$-split) we have that
  $\phi \equiv \phi \lor \phi$ if, and only if, $\phi$ is union closed.
\end{lemma}
\begin{proof}
  Every formula $\phi$ implies $\phi\lor\phi$ because of the empty multiteam property.
  Assume that $\phi$ is union closed.
  Let $\fA$ be a structure and $M$ a multiteam such that $\fA \models_M \phi \lor \phi$.
  Thus there is a split  $M = S \uplus T$ with $\fA \models_S \phi$ and $\fA \models_T \phi$.
  By union closure we obtain $\fA \models_{S \uplus T} \phi$.
  For the other direction assume that $\phi \equiv \phi \lor \phi$ and that $\fA \models_S \phi$ and $\fA \models_T \phi$ hold.
  Therefore $\fA \models_{S \uplus T} \phi \lor \phi$ and hence $\fA \models_{S \uplus T} \phi$.
\end{proof}

The following example illustrates that this postulate is violated if we interpret
disjunction by ($\lor$-cover).

\begin{example}
  \label{ex:cover disjunction not idempotent}
 Let $\phi = x \mincl y$ which is union-closed. Consider the multiteam 
 $M=\mset{xy\mapsto 01, xy\mapsto 10, xy\mapsto 10}$  with $\{0,1\} \nmodels_M \phi$.
However, if we would define disjunction by \eqref{eq:cover disjunction}
 we could take $S=\mset{xy\mapsto 01, xy\mapsto 10}$ with
 $S\subsetpeq M \subsetpeq 2\cdot S$
 and  $\{0,1\} \models_{S} \phi$, which would imply that
  $\{0,1\} \models_M \phi \lor \phi$.
\end{example}

\subsection*{Existential Quantification}
To define the semantics of a formula $\exists x \phi$ on a multiteam $M$
(and a structure $\fA$) we have to
provide an adequate definition  of how $M$ should be modified by values for
the quantified variable $x$ to obtain a multiteam that satisfies $\phi$. 
Thus the definition should have the form that
$\fA \models_M \exists x\phi$  if, and only if,  $\fA \models_{M[x \mapsto F]} \phi$
for some function $F\colon D \to C$.

The domain $D$ of $F$ can either be the multiteam $M$ 
(so that $F$ can assign different values for $x$ to 
different occurrences of an assignment in $M$) or
its underlying team $\supp{M}$ (so that all occurrences of
an assignment get the same value for $x$).
Also for the codomain $C$ of $F$ there are several possible choices:
It can either be the universe $A$ of $\fA$ (so that
for each assignment we assign precisely one value to $x$), which is our preferred choice,
or the set $\ppot{A}$ (which, as in team semantics, means that
we update each assignment by a non-empty set of values for $x$).
Further, we could even assume that $F$ assigns to each $s\in M$ a multiset of $k$ values for $x$,
for some fixed positive $k$ (constant for all assignments).
This last possibility is of a probabilistic flavour, similar to the setting of probabilistic team semantics
\cite{DurandHanKonMeiVir18} or \cite{HyttinenPaoVaa15}, and removes, for instance, the information on
how many assignments had originally been present. This is not the right choice for the setting
multiteam semantics intended to argue with multiplicities, and would also not be in line with 
our definition for the disjunction, which is not probabilistic in nature.
A further problem with such a choice concerns the algorithmics of evaluation: 
it is not clear at all that a similar result to Corollary
\ref{cor:NP} would hold if the existential quantifier could increase the size of
multiteams arbitrarily.

Altogether this gives four remaining possible definitions but it turns out that
only one of these is compatible with our postulates.
Taking the codomain $C=\ppot{A}$, which is the right choice
in team semantics, leads in multiteam semantics to a violation of
the postulate that a dummy quantification of a variable that is not present in
an atomic dependence should not change its meaning.
To put it differently, this principle means that for a multiteam $M$ with domain $\cV$
and a variable $x\not\in\cV$ the multiteam
$M[x\mapsto F]\res{\cV}$ satisfies the same atomic dependencies as $M$ does.
However, if we permit the sets $F(s)$ to be arbitrary non-empty subsets of $A$,
then these may have different sizes, so the multiteam
$M[x\mapsto F]\res{\cV}$ may have completely different relative multiplicities
of its assignments, and therefore different atomic dependence properties.
As an example, consider $M = \mset{s_0\colon xy \mapsto 01, s_1\colon xy\mapsto 10}$ and $F\colon s_0 \mapsto \{0\}, s_1\mapsto\{0,1\}$.
Since $M = \mset{s_0, s_1}$ and $\supp{M} = \{s_0, s_1\}$ it is irrelevant for
this particular example how we define the domain of our choice functions.
We have $\{0,1\} \models_M x\mincl y$ but $\{0,1\} \nmodels_{M[z\mapsto F]} x\mincl y$.
These considerations show that the function $F$ must assign to each occurrence of an assignment $s$ precisely one element of $A$.

It remains to show that  a choice function $F$ should have the
liberty to assign different values for $x$ to
different occurrences of an assignment in $M$.
For this let $\fA = \{a,b\}$ and $M = 2\cdot\mset{x\mapsto a}$ and consider the
formula $\phi \ceq \E y\A z(xz\mincl xy)$. It states, for $\fA$ and $M$,  that we can assign to $y$ values such that
for all values $a$ of $x$ (which in our case is just one) the multiset $M[y\mapsto F]\res{x=a}(y)$ is a multiple of $A$ (cf.~Proposition \ref{prop: mincl basic properties}).
For a choice function that can assign different values to different occurrences of an assignment, we certainly have that $\fA\models_M\phi$.
But if we forbid this liberty then this would not be the case.
However, consider the multiteam $M' = \mset{s_1\colon xu\mapsto aa,s_2\colon xu\mapsto ab}$. 
In this case, the restricted choice function $G\colon s_1 \mapsto a, s_2\mapsto b$ satisfies $\fA\models_{M'[y\mapsto G]}\A z(xz\mincl xy)$ although $M = M'\res{x}$.
Thus, admitting only restricted choice functions, that must assign the same value to each occurrence of
an assignments would lead to a violation of the locality principle.

\subsection*{Universal versus existential quantification}
Finally, we discuss a possible postulate that we had considered at some point, but then decided to discard,
and that is not satisfied by the multiteam semantics that we propose here. In most common
logics it is obviously true that $\A x\phi \models \E x\phi$.
However, we do not think that it is justified to demand this as a postulate in our context, because of the
different roles that these quantifiers have in the setting of teams and multiteams,
compared to classical Tarskian settings. Classically, a universal quantifier can be thought of
as giving an opponent the power to \emph{choose} an element that makes `proving'
the statement at hand as difficult as possible, whereas an existential quantifier gives
this choosing power to the proponent, so that she can select an element that makes the
`proof' as simple as possible. However, in a setting of teams and multiteams, a universal quantifier does
not represent a choice, but a \emph{uniquely defined extension} of the current semantic object that
involves all elements of the given structure. The existential quantifier instead does indeed represent
a choice. Therefore we see in this context no specific reason to impose that a universal quantification should
generally imply  an existential one. Of course there can still be settings where this is the case, such as in
lax team semantic, because the unique universal extension of a team happens to be also a Skolem extension,
or in probabilistic team semantics.
However, the implication $\A x\phi \models \E x\phi$ does not hold, for instance,  in strict team semantics.
In our setting, the universal quantifier increases the size of a multiteam while the existential one (as defined here) does not.
One example of a formula $\phi$ for which $\A x \phi \models \E x\phi$ fails,  is $\phi(x):=\fork{\leq1/2}x$.
Indeed, $\fA\models_{\mset{\emptyset}}\A x\;\fork{\leq1/2}x$ holds for all structures with at least two elements,
but $\fA\nmodels_{\mset{\emptyset}} \E x\;\fork{\leq1/2}x$.

\section{Comparison of logics with different multiteam dependencies}
\label{sec: comparison logics with mts}

As in team semantics, the logical operators \emph{preserve} certain fundamental properties of formulae,
in the sense that whenever both $\psi$ and $\phi$ have that property, then so do
$\psi\lor\phi$, $\psi\land\phi$, $\E x\psi$ and $\A x\psi$. Examples of such properties are:

\begin{description}
  \item[Downwards closure:] $\fA \models_M \psi$ implies $\fA \models_N \psi$ for all $N\subsetpeq M$.
  \item[Team semantical downwards closure:] $\fA \models_{(X, m)} \psi$ implies $\fA \models_{(Y, m\res{Y})} \psi$ for all $Y \subseteq X$.
  \item[Union closure:] $\fA \models_R \psi$ and $\fA \models_S \psi$ implies $\fA \models_{R\uplus S} \psi$.
  \item[Closure under scalar multiplication:] $\fA \models_M \psi$ implies $\fA \models_{kM} \psi$ for all $k\in\N$.
  \item[Validity in the empty multiteam:] $\fA \models_{(\emptyset, \emptyset)} \psi$.
\end{description}

From the closure properties of atomic multititeam dependencies, we thus obtain the following closure properties for logics with 
multiteam semantics.

\begin{prop}
  \label{prop: mts closure}
  \begin{itemize}
    \item All formulae in $\fom[\dep,{}\excl{}]$ are downwards closed.
    \item All formulae in $\fom[\fork{\geq q}, \fork{> q}]$ are team semantically downwards closed.
    \item All formulae in $\fom[\mincl,\fork{\leq q}, \fork{< q}]$ are union closed.
    \item All formulae in $\fom[\emptyset]$ are flat, i.e.\ $\fA\models_M \phi$ if, and only if, $\fA\models_{\mset{s}} \phi$ for all $s\in M$.
    \item All formulae in $\fom[\dep,{}\excl{},\mincl, \fork{*},\mindep]$ have the empty multiteam property, and are closed under scalar multiplication.
  \end{itemize}
\end{prop}

Recall that for logics with team semantics, the relationship between logics with different atomic dependencies,
in terms of expressive power, is well-understood (see \cite{Galliani12,GraedelVaa13}).
\begin{itemize}
\item $\fot \precneqq \fot[\dep] \equiv \fot[{}\excl{}]\precneqq\fot[\incl,{}\excl{}]\equiv \fot[\mindep]\equiv\fot[\mindep_c]$.
\item $\fot \precneqq \fot[\anon] \equiv \fot[\incl] \precneqq\fot[\incl,{}\excl{}]$. Further $\fot[\incl]$ and $\fot[{}\excl{}]$ are incomparable.
\end{itemize}

We would like to understand the relative expressive power of logics with different multiset dependencies in a similar way.
Some of the results from team semantics carry over (with minor modification of the proofs) to multiteam semantics.

\begin{prop}
  \label{prop: mts relationship}
  \begin{itemize}
    \item $\fom[\dep] \equiv \fom[{}\excl{}] \precneqq \fom[\fork{\geq1/2}] \precneqq \fom[\fork{=1/2}] \precneqq \fom[\mindep]$.
    \item $\fom[\mincl_\alpha] \preceq \fom[\mincl] \equiv \fom[\fork{\leq1/2}] \precneqq \fom[\fork{=1/2}]$ for every $\alpha\in\fo$.
    \item There  are $\alpha,\beta\in\fo$ such that $\fom[\mincl_\alpha] \equiv \fom[\mincl]$ and $\fom[\mincl_\beta] \equiv \fo^M[\emptyset]$.
\end{itemize}
\end{prop}
\begin{proof}
  The team semantical translations between exclusion and dependence logic from \cite{Galliani12}
  work in multiteam semantics as well.
  Translations proving the claims made for forking atoms can be found in \cite{Wilke20}.
  For the different variants of inclusion, observe that $\tx\mincl_\true\ty \equiv \tx\mincl\ty$ and $\tx\mincl_\false\ty \equiv \true$,
  and that $\tx\mincl_\alpha\ty \equiv \E\tx'(\tx'\mincl\tx \land (\alpha(\tx')\imp\tx'=\ty))$.
  
  Let us explain the last equivalence. For a multiteam $M$ on variables $\tx$ and $\ty$, let
  $M_\alpha=\mset{s\in M : M\models\alpha[s]}$.
  The formula $\tx\mincl_\alpha\ty$ is true in $M$ if, and only if, there is
  an injective function $f\colon M_\alpha\to M$ such that $(fs)(\ty)=s(\tx)$ for all $s\in M_\alpha$.
  Now choose any bijection $g\colon M\to M$ that extends $f^{-1}$, i.e.~$(g\circ f)(s)=s$ for
  $s\in M_\alpha$. The Skolem extension $M[\tx'\mapsto F]$ where $F\colon M\ra A^k$ extends
  any $t\in M$ to $t[\tx'\mapsto(gt)(\tx)]$ clearly satisfies $\tx'\mincl \tx$ (because $g$ is a bijection).
  Further, if an assignment $t[\tx'\mapsto(gt)(\tx)]$ satisfies $\alpha(\tx')$ then the
  assignment $gt$ satisfies $\alpha(\tx)$; hence $gt\in M_\alpha$ which means that
  $t(\ty)=(gt)(\tx)$, and thus implies that $t[\tx'\mapsto(gt)(\tx)]$ satisfies $\tx'=\ty$.
  For the converse, assume that for some tuple of values $\tc$ such that $\alpha(\tc)$ holds,
  there are more assignments $s\in M$ with $s(\tx)=\tc$ than assignments $t$ with $t(\ty)=\tc$.
  Then any Skolem extension $M[\tx'\mapsto F]$ that satisfies $\tx'\mincl \tx$, and hence also
  $\tx\mincl\tx'$ must have more assignments $s$ with $s(\tx')=\tc$ than assignments with
  $s(\ty)=\tc$, so the formula on the right hand side does not hold for this multiteam.
\end{proof}

Somewhat deeper arguments are needed to show that multiteam independence logic $\fom[\mindep]$ is powerful enough to express all of  $\fom[\mincl]$.
The proof is based on ideas from \cite{Galliani12}.
A similar result for probabilistic team semantics was established in \cite{HannulaHirKonKulVir19}.

\begin{prop}
  \label{prop: mincl in mindep}
  $\fom[\mincl] \precneqq \fom[\mindep]$.
  Especially, $\tx \mincl \ty$ is equivalent to
  \begin{align*}
    \phi_{\mincl}(\tx, \ty) \ceq \A\tz\A a\A b \big(
    & (\tz \neq \tx \land \tz \neq \ty) \lor (\tz \neq \ty \land a \neq b) \lor\\
    & (\tz   =  \ty \land a   =  b) \lor ((\tz = \ty\lor a=b) \land \tz\mindep ab)\big).
  \end{align*}
\end{prop}
Let us explain the formula $\phi_{\mincl}=\A \tz,a,b(\phi_1\lor\phi_2\lor\phi_3\lor\phi_4)$
on an intuitive level before we formally verify that it actually states what is claimed here.
The variable tuple $\tz$ is used as a bridge between the values of $\tx$ and $\ty$,
and we are not interested in values of $\tz$ that do not coincide with either
of both and remove these using the first disjunct $\phi_1$.
The actual check, whether all values for $\tx$ and $\ty$ occur equally often happens
in the final subformula $\phi_4$, but we have to prepare the values we want to pass
to it.
The reasoning behind the first-order part of $\phi_4$ (that it only accepts
assignments in which $a$ and $b$ agree, or if $\tz$ assumes the same value as $\ty$
does) is that we have to verify that every value of $\tx$ occurs at least as often
for $\ty$, hence, whenever $\tz$ equals $\ty$ we are at a \quotes{good} value and
thus put no restriction on $a$ and $b$.
But if we see an assignment in which $\tz$ agrees with $\tx$ we enforce that $a$ and $b$
coincide.
Accordingly, if a value $\tu$ occurs more often for $\tx$ than for $\ty$ it must
be the case that the probability that $\tz$ assumes $\tu$ is larger if one knows
that in fact $a$ and $b$ agree, which violates the independence atom $\tz\mindep ab$.
The subformula $\phi_2$ simply helps us to eliminate the assignments that are
forbidden in $\phi_4$, while $\phi_3$ helps us to get rid of those assignments
in which $\tz=\ty\neq\tx$ and $a=b$ holds.
This is necessary because every assignment satisfying $\tz=\tx\neq\ty$ must be
distributed to $\phi_4$ wherefore we already have assignments in which $a$ and
$b$ agree.
But the assignments in which $\tz$ assumes the same value, but this time agrees
with $\ty$ are necessary to \quotes{counter} these and hence equally often must
map $a$ and $b$ to different values.
The subformula $\phi_3$ takes its role by accepting such assignments where $a=b$.
Notice that in the formula given by Galliani for team semantics a similar formula like $\phi_3$ is not necessary because it only matters that for every
value that $\tz$ might assume for all combinations of values for $a$ and $b$ there
are assignments that actually take these values, but their multiplicity does not count.

\begin{proof}[Proof of Proposition \ref{prop: mincl in mindep}]
  Assume $\fA\models_M\tx\mincl\ty$ and let $M'$ be the multiteam after assigning all values
  to the variables $\tz,a$ and $b$, i.e.~$M'=M[\tz\mapsto A^{|\tz|}][a\mapsto A][b\mapsto A]$.
  To prove that indeed $\fA\models_{M'} (\phi_1\lor\phi_2\lor\phi_3\lor\phi_4)$ let
  $M_1 := M'_{\phi_1}$ and $M_2 := M'_{\phi_2}$.
  Recall that $M_\phi$ is the submultiset of all assignments in $M$ that satisfy $\phi$.
  Therefore, all remaining assignments $s \in M'' := M' \setminus (M_1 \uplus M_2)$
  either satisfy $\tz = \ty$ or $\tz = \tx \neq \ty \land a=b$.
  Assignments of the first kind can be divided further based on whether they
  satisfy $\tx=\ty$: put all assignments of $M''$ that do not obey this equation
  and moreover agree on $a$ and $b$ into $M_3$, i.e.~$M_3 := M''_{\tx\neq\ty\land a=b}$,
  and let $M_4 := M''\setminus M_3$.
  It follows immediately from the definitions of $M_i$ that
  $\fA \models_{M_i} \phi_i$ for $i = 1,2,3$.
  What is left to show is that indeed $\fA \models_{M_4} \phi_4$ holds.
  The first-order part, that is $\tz = \ty\lor a=b$, of $\phi_4$ is satisfied by
  $M_4$ since every assignments $s\in M''$ satisfies it and $M_4 \subsetpeq M''$.
  We claim that for all values which $\tz$ may assume in $M_4$ every value combination
  for $a$ and $b$ occurs equally often in $M_4$.
  \begin{claim}
    \label{claim:M_4}
    For all $\tc \in \supp{M_4}(\tz)$ and $d,e,f,g\in A$ with
    $\hat{M} := (M_4)_{\tz=\tc}$ we have $|\hat{M}_{ab=de}| = |\hat{M}_{ab=fg}|$.
  \end{claim}
  \begin{claimproof}
    Split the assignments in $\hat{M}$ according to whether they satisfy $\tx=\ty$
    into $\hat{M}_{\tx=\ty}$ and $\hat{M}_{\tx\neq\ty}$.
    By construction of $M_4$ we conclude that $\hat{M}_{\tx=\ty}$ can be written as
    $\tilde{M}[a\mapsto A][b\mapsto A]$ for some $\tilde{M}$, since we have
    distributed every $s\in M''$ with $s(\tx) = s(\ty)$ into $M_4$.
    Regarding $\hat{M}_{\tx\neq\ty}$, we observe that the number of assignments
    mapping $\tx$ to $\tc$ is the same as the number of those mapping $\ty$ to $\tc$
    because by assumption we know that $\fA \models_M\tx\mincl\ty$ holds.
    The assignments of the first kind must map $a$ and $b$ to identical values,
    while the assignments of the second kind must map $a$ and $b$ to different
    values.
    But which values those variables take is in no other way restricted and hence
    they occur equally often for all possibilities $d,e,f$ and $g$.
  \end{claimproof}
  From this claim it immediately follows that $\fA\models_{M_4}\tz\mindep ab$,
  which concludes the proof of the left-to-right direction.
  
  For the converse direction let $\fA \models_M \phi_{\mincl}$ and let $M'$ be as above.
  That means $\fA \models_{M'} (\phi_1\lor\phi_2\lor\phi_3\lor\phi_4)$ holds
  and hence $M' = M_1\uplus M_2 \uplus M_3 \uplus M_4$ such that $\fA \models_{M_i}\phi_i$
  for $i=1,2,3,4$.
  Towards a contradiction, assume that $\fA\nmodels_M\tx\mincl\ty$, thus a value
  $\tc\in \supp{M}(\tx)$ exists such that $|M_{\tx=\tc}| > |M_{\ty=\tc}|$.
  First, notice that all assignments in $M'$ which satisfy $\tz=\tx\neq\ty\land a=b$
  must belong to $M_4$ as they falsify all other formulae $\phi_i$ for $i=1,2,3$.
  Consider the submultiteam $\hat{M}$ of $M_4$ that consists of all assignments $s$
  satisfying $s(\tz) = \tc$.
  We can divide $\hat{M}$ into the two multiteams $\hat{M}_{\tz=\ty}$ and $\hat{M}_{\tz\neq\ty}$.
  Because $|M_{\tx=\tc}| > |M_{\ty=\tc}|$ it follows that also $|\hat{M}_{\tz\neq\ty}|
  > |\hat{M}_{\tz=\ty}|$.
  But all assignments in $\hat{M}_{\tz\neq\ty}$ must agree on $a$ and $b$, as
  otherwise the first-order part of $\phi_4$ would not be satisfied by $M_4$.
  At this point we may safely assume that there are two different values $d\neq e$ in $A$,
  as otherwise all multiteams with base $A$ must satisfy the inclusion $\tx\mincl\ty$.
  Accordingly, in $\hat{M}$, the probability that $\tz=\tc$ holds is larger given
  the fact that $a=b=d$ than if one knows $a=d\neq e=b$, i.e.~$\pr_{\hat{M}}(\tz=\tc\mid ab=dd)
  > \pr_{\hat{M}}(\tz=\tc \mid ab=de)$.
  This projects back to $M_4$, because $\hat{M}$ is just the collection of all
  assignments $s$ in $M_4$ with $s(\tz) = \tc$.
  Thus, $\fA\nmodels_{M_4}\phi_4$ and hence $\fA\nmodels_M\phi_{\mincl}(\tx,\ty)$.
  
  The inclusion is proper because $\fom[\mincl]$ is union closed, while $\fom[\mindep]$ is not.
\end{proof}

However, it is by no means the case that all connections between logics with team semantics carry over to the
setting of multiteam semantics.
A very interesting issue is the relationship between independence logic on one side,
and inclusion-exclusion logic on the other side.  The fact that $\fot[\incl, {}\excl{}]$ has the full expressive power
of independence logic $\fot[\indep]$ and, in fact of any logic $\fot[\Omega]$ based on atomic dependencies
that are first-order (or even existential second-order) definable \cite{Galliani12} has been of fundamental importance for
many results in team semantics. We show that analogous statements for multiteam semantics fail, in a rather strong sense,
and this applies not just to inclusion and exclusion, but to any combination of downwards closed and union closed atomic dependencies.

\begin{theorem}
  \label{thm:down + union != mindep}
  Let $\talpha$ be any collection of downwards closed atoms and $\tbeta$ be any collection of union closed atoms.
  There is no formula $\psi \in \fom[\talpha, \tbeta]$ with  $x\mindep y \equiv \psi(x, y)$.
\end{theorem}
\begin{proof}
For every natural number $k$, consider the multiteam
\begin{align*}
  M_k:=  &(\mset{x\mapsto0}\uplus k\mset{x\mapsto1})\times(\mset{y\mapsto0}\uplus k\mset{y\mapsto1})\\
  = &\mset{xy\mapsto00}\uplus k\mset{xy\mapsto01,xy\mapsto10}\uplus k^2\mset{xy\mapsto11}.
\end{align*}

\noindent
Notice that $x\mindep y$ holds in $M_k$  for all $k\geq 1$. Further, a simple calculation shows that
$x\mindep y$ holds in $M_k\uplus M_\ell$ if, and only if,  $k=\ell$.
Towards a contradiction we assume that there is a formula $\psi(x,y)\in\fom[\talpha,\tbeta]$ equivalent to $x\mindep y$.
 
Let $\cT(\psi)$ be the syntax tree of $\psi$ and choose, for every $k$, an annotation $F_k$ of $\cT(\psi)$,
which associates with every subformula $\phi\in\cT(\psi)$ a multiteam $M_k(\phi)$ over $\{0,1\}$, and which
witnesses that $M_k$ satisfies $\psi$.  Further, let $T(F_k)$ be the restriction of $F_k$ to
the supporting teams, mapping $\phi$ to $T(M_k(\phi))$.
Although there are possibly infinitely many different 
annotations $F_k$, there are only finitely many different restrictions $T(F_k)$. Hence there exist $k, \ell$ with
$k\neq \ell$ and $T(F_k)=T(F_\ell)$.
The union $F_k\uplus F_\ell$ is therefore a structurally valid candidate for witnessing that
the multiteam $M_k\uplus M_\ell$ satisfies  $\psi$.
It remains to show that each atomic subformula $\phi$ in $\psi$ is satisfied by the the multiteam $M_k(\phi)\uplus M_\ell(\phi)$.
For union closed atoms $\phi=\beta_i$ this is clear because $\phi$ holds in $M_k(\phi)$ and in $M_\ell(\phi)$.
For atoms $\phi=\alpha_j$ and for first-order literals $\phi$ we use the fact that
the underlying team of $M_k(\phi) \uplus M_\ell(\phi)$ is the same as the one for $M_k(\phi)$,
and since, by Proposition~\ref{prop: downward atom = ts atom}, $\phi$ only depends on the team, we have that
the multiteam $M_k(\phi)\uplus M_\ell(\phi)$ satisfies $\phi$.
Putting everything together we have that $M_k\uplus M_\ell$ satisfies $\psi(x,y)$ but not $x\mindep y$, hence $\psi(x,y)\not\equiv x\mindep y$.
\end{proof}

\begin{cor}
  $\fom[\mincl, {}\excl{}]\precneqq\fom(\mindep)$.
\end{cor}

\subsection*{An open problem:} In team semantics, and in fact also in probabilistic team semantics \cite{DurandHanKonMeiVir18a}, 
independence logic is equivalent to conditional independence logic, but the 
proofs do not translate to the multiteam setting and
the precise relationship between independence and conditional independence under multiteam semantics remains open.

\section{Back \& Forth between Team and Multiteam Semantics}
\label{Sect: BandF}

We next discuss the relationship between team semantics and multiteam semantics.
As already pointed out, any team semantical atom $\alpha\in\fot$ can be understood naturally as a multiteam atom
such that $\fA\models_{(X, n)}\alpha$ if, and only if, $\fA\models_X\alpha$.
Conversely, we would like to associate with every multiteam atom $\beta$ a team semantical atom $\teamVer{\beta}$.
We stipulate that $\teamVer{\beta}$ should be the strongest team semantical atom that is implied by $\beta$.
This means, on one side that $\fA\models_{(X, n)}\beta$ implies that $\fA\models_{X}\teamVer{\beta}$,
and on the other side that $\teamVer{\beta}\models\gamma$ for any other atom $\gamma$ satisfying $\fA\models_{(X, n)}\beta$ implies $\fA\models_X\gamma$.
These rules imply the following semantics for $\teamVer{\beta}$.

\begin{definition}
  Given a multiteam atom $\beta$, let $\teamVer{\beta}$ be the team semantical atom with
  $\fA\models_X \teamVer{\beta}$ if, and only if, $\fA\models_{(X, m)} \beta$ for some $m\colon X\to\N_{>0}$.
\end{definition}

\begin{example}
 It is easily verified that $\teamVer{(\tx\mindep\ty)} \equiv \tx\indep\ty$ and
$\teamVer{(\tx\fork{\leq \frac{1}{2}} y)} \equiv \tx\anon y$.
The team semantical version of a multiteam inclusion atom $\tx\mincl \ty$ is less obvious.
We have that $\teamVer{(\tx\mincl\ty)}$ is neither $\tx\incl\ty$, nor $\tx\equiex\ty$.
Indeed, the team $\{xy\mapsto00,xy\mapsto01,xy\mapsto11\}$ satisfies both $x\incl y$ and $x\equiex y$, but no extension to a multiteam satisfies $x\mincl y$.
To define such an atom we introduce an auxiliary concept:
Let $X$ be a team and $\tx, \ty \subseteq \dom(X)$ be variable tuples of the same length.
The graph $\teamG{X}{\tx}{\ty}$ has vertex set $X(\tx)\cup X(\ty)$
and an edge from $\ta$ to $\tb$ if there exists an assignment $s\in X$ with 
$s(\tx) = \ta$ and $s(\ty) = \tb$.
Further, let $\cycle{\tx}{\ty}$ be a team semantical atom satisfied by $X$ if the graph
$\teamG{X}{\tx}{\ty}$ is decomposable into (not necessarily disjoint) cycles.
This means that $\teamG{X}{\tx}{\ty} = C_1 \cup\dots\cup C_k$ for  cycles $C_i$, where the
union of two graphs $G = (V, E)$ and $H = (W, F)$ is defined as $G\cup H :=
(V\cup W, E\cup F)$ ( so they may share both vertices and edges).
By item (\ref{prop: mincl basic properties - decompose}) of Proposition \ref{prop: mincl basic properties}
it then follows that
$\teamVer{(\tx\mincl\ty)} \equiv \cycle{\tx}{\ty}$.
\end{example}

For any collection $\Omega$ of multiteam atoms, and the corresponding collection $\teamVer{\Omega}:=\{\teamVer{\beta} : \beta\in\Omega\}$
we obtain a mapping that associates with every formula $\psi\in\fom[\Omega]$ a corresponding formula $\teamVer{\psi}\in\fot[\teamVer{\Omega}]$.
For discussing connections between multiteam and team semantics, we introduce the following notions.

\begin{definition}
  \label{def: team semantical version}
  Let $\psi\in\fom$ and $\phi\in\fot$.
  We say that:
  \begin{itemize}
    \item $\psi$ \emph{depends only on the team} if for all $\fA$ and all multiteams $M,N$ with $\supp{M}=\supp{N}$ we have that
    $\fA\models_M  \phi$ if, and only if, $\fA\models_N \psi$;
    \item $\phi$ is a \emph{team semantical version} (\tsv) of $\psi$ if for all $\fA$ and $X$ we have that
    $\fA\models_X\phi$ if, and only if, $\fA\models_M \psi$ for some multiteam $M$ with $\supp{M}=X$;
    \item $\phi$ and $\psi$ are \emph{companions} if $\phi$ is a \tsv\ of $\psi$ and $\psi$ depends only on the team.
  \end{itemize}
\end{definition}

With this terminology, we can restate our second postulate as follows.
For any collection $\Omega$ of team semantical atoms, every formula
$\phi\in\fot[\Omega]$ is a team semantical version of itself, understood
as a formula in $\fom[\Omega]$.

\begin{prop}
  Let $\psi\in\fom$ be a multiteam formula, $\beta$ be a multiteam dependency atom and $M$ be a multiteam.
  \begin{enumerate}
  \item $\teamVer{\beta}$ is a team semantical version of $\beta$.
  \item Whenever $\fA\models_M\psi$ then $\fA\models_{\supp{M}} \teamVer{\psi}$.
  \item If $\psi$ depends only on the team then it has a companion $\phi$.
  \end{enumerate}
\end{prop}
\begin{proof}
The first two claims follow directly from definition of
  $\teamVer{\beta}$  by structural induction on the formulae.
 Towards (3) we notice that for any multiteam $M$ with $X=\supp{M}$ we have that
  $\fA\models_M \psi$ if, and only if, $\fA\models_{(X, s\mapsto1)}\psi$, which, for fixed $\psi$, is an $\NPtime$
  property of $\fA$ and $X$ (cf.~Corollary \ref{cor:NP}).
  Hence, by Fagin's Theorem, there exists a $\eso$-sentence $\theta$ such that
  $\fA\models_{(X, s\mapsto1)}\psi$ if, and only if, $(\fA, X) \models \theta$.
  But since inclusion-exclusion logic $\fot[\incl,{}\excl{}]$ is equally expressive
  as $\eso$ \cite{Galliani12}, there also exists a formula $\phi\in\fot[\incl,{}\excl{}]$ such that
  $(\fA, X) \models \theta$ if, and only if, $\fA\models_X\phi$.
\end{proof}

We further notice that the converse of (2) is not true, so in general
$\teamVer{\psi}$ need not be a  \tsv\ of $\psi$.
To see this consider the formula $\psi = x\mincl y \land u\mincl w$.
It is easy to see that $\teamVer{\psi}$ is satisfied by the team $X = \{xyuw\mapsto0101,xyuw\mapsto1001,xyuw\mapsto1010\}$.
However, there is no function $n\colon X\to\N_{>0}$ so that the multiteam $(X,n)$ satisfies $\psi$.

By Proposition~\ref{prop: downward atom = ts atom}, we have that \emph{downwards closed} atoms $\beta\in\fom$
depend only on the team, so $\beta$ and $\teamVer{\beta}$ are companions. This extends to $\fom$ in the following sense.

\begin{prop}
  \label{prop: psi dc => psi^T and psi companions}
  Let $\psi\in\fom[\Omega]$ for any collection $\Omega$ of downwards closed atoms.
  Then $\psi$ and $\teamVer{\psi}$ are companions.
\end{prop}

The following theorem exhibits a stronger relation between downwards closed properties and those that depend only on the team.

\begin{theorem}
  \label{thm: mts[dep, mincl] dep on team iff dc}
  A formula $\psi\in\fom[\dep, \mincl]$ depends only on the team if, and only if, it is downwards closed.
\end{theorem}
\begin{proof}
  The proof of the direction from right to left follows the argument given in
  Proposition \ref{prop: downward atom = ts atom} and uses the fact that all multiteam
  atoms are assumed to be closed under scalar multiplication and that the logical
  connectives preserve this property (cf.~Proposition \ref{prop: mts closure}).
  
  Assume that $\psi$ depends only on the team and that $\fA\models_{(X, n)}\psi$.
  For arbitrary $s\in X$ we show that also $\fA\models_{(X\setminus\{s\}, n)}\psi$.
  Let $M_i$ be the multiteam with base $X$ in which all assignments except $s$ have multiplicity $i$ and $s$ occurs just once. By the assumption,  $\fA\models_{M_i}\psi$ for all $i>0$.
  For each $i$ we fix an annotation $\cT_i$ of the syntax tree of $\psi$ witnessing that $\fA\models_{M_i}\psi$.
  Notice that every assignment $s'\in X$ gives rise to a layer in such an annotation, that is the restriction of $\cT_i$ to  just the assignment $s'$ (in the root) and all of its extensions (universal or existential).
  Moreover, note that only a finite number (depending on $X$, $\fA$ and $\psi$) of such restrictions exist.
  Denote them by $T_1,\dots,T_\ell$.
  Hence every $\cT_i$ is decomposable into $a^i_1\cdot T_1 \uplus \dots \uplus a^i_\ell\cdot T_\ell$ for $a^i_j\in\N$, where $n\cdot\cT$ means the annotated syntax tree in which every multiteam is multiplied by the factor $n$.
  By the pigeonhole principle, there exist two annotations $\cT_i$ and $\cT_j$ such that regarding their decompositions we have $a^i_1 \leq a^j_1, \dots, a^i_\ell\leq a^j_\ell$, where some inequality is strict. 
   In other words, $\cT_i(\theta) \subsetpeq \cT_j(\theta)$ for all subformulae $\theta$ of $\psi$.
  We can construct now the tree $\cT \ceq \cT_j\setminus\cT_i$, that is defined in the obvious way, by $\cT(\gamma) \ceq \cT_j(\gamma) \setminus \cT_i(\gamma)$. Clearly, $\cT$  is a structurally valid annotation and thus a possible 
  witness for $\fA\models_{(X\setminus\{s\}, k)}\psi$.
  We claim that, indeed, all atomic formulae are satisfied in $\cT$.
  This is clear for all downwards closed atoms (and in particular all first-order literals), since they are satisfied by $\cT_j$ already.
  For multiteam inclusion atoms $\beta$ we conclude from $\cT_i(\beta) \models \beta$, $\cT_j(\beta)\models\beta$, $\cT_i(\beta)\subsetpeq\cT_j(\beta)$ and Proposition \ref{prop: mincl basic properties} that also $\cT(\beta) \models \beta$ holds.
\end{proof}

\section{Multiteam semantics and existential second-order logic}
\label{sect:eso}

A common way for describing the expressive power of a logic with team semantics is to relate it
to a fragment of existential second-order logic with classical Tarski semantics. 
The question arises whether there are similar results for logics with multiteam semantics. 
Since multiteam semantics involves reasoning about multiplicities of assignments, it must be related
to a framework of classical logic that includes arithmetical operations on natural numbers, but 
should, at the same time, keep these separate from the objects in the structures and assignments under consideration.
Such frameworks have been developed, in a very general sense, in \emph{metafinite model theory} \cite{GraedelGur98}.
In general a metafinite structure is a triple $\fD=(\fA,\fN,W)$ where $\fA$ is a finite structure, $\fN$ is a (typically
infinite) numerical structure, possibly equipped with \emph{multiset operations}, and $W$ is a set of weight functions
$w\colon A^k\ra N$ mapping tuples from $\fA$ to numbers in $\fN$. 
Here we use a specific variant of metafinite structure where the numeric part 
is the standard model of arithmetic $\fN=(\N,0,+,\,\cdot\,,\Sigma)$, together with the summation 
operator over arbitrary finite multisets (mapping every finite multiset $M$ of natural numbers to the sum $\Sigma(M)$ of 
its elements).

\begin{definition}
  A \emph{multiteam structure} is a triple $\fD=(\fA,\fN,\{f\})$ where $\fN$ is as above and $f\colon A^k\to \N$ is a function, representing a multiteam $M_f=(X,m)$ with domain $x_1,\dots,x_k$, by $f(\ta)=m(s)$ if the assignment $s\colon \tx\mapsto\ta$ is in $X$, and $f(\ta)=0$ otherwise.
\end{definition}

In metafinite model theory, first-order logic is defined in the standard way, but with the
important stipulation that \emph{variables and quantifiers range only over elements of $\fA$}. 
Here, we further require that comparisons of terms $t=t'$ may only be used positively.
Thus, terms and formulae of first-order logic for multiteam structures, denoted $\fomts[+,\,\cdot\,]$, 
is defined by 
\[
  t  \cceq 0 \mid f\tx \mid t+ t\mid t\cdot t \mid \Sigma_{\tx} \{ t : \alpha\}\quad\text{ and }\quad
  \phi \cceq \alpha \mid
    t = t \mid
    \phi\land\phi \mid \phi\lor\phi \mid \E x\phi\mid  \A x\phi,
\]
where $\alpha$ are arbitrary literals in the vocabulary of $\fA$.
The semantics is the usual Tarski semantics. For each multiteam structure   
$\fD=(\fA,\fN,\{f\})$, and each assignment $s\colon\bar x\mapsto\bar a$,
it assigns to each term $t(\bar x)$,  a value $t^\fD[\bar a]\in\N$, and
defines, for each formula $\phi(\bar x)$, whether or not $\fD\models\phi[\bar a]$.
Everything is completely standard except perhaps the use of summation as a multiset operation.
Given a term $t(\bar x,\bar y)$ and a literal $\alpha(\bar x,\bar y)$,
the term $\Sigma_{\tx} \{ t : \alpha\}$ binds $\bar x$
but leaves $\bar y$ free. Given any assignment $\bar y\mapsto\bar b$, we obtain the
multiset $M_t[\bar b]:=\mset{ t^\fD[\bar a,\bar b]:  \bar a\in A^k, \fD\models\alpha[\bar a,\bar b]}$.
The value of  $\Sigma_{\tx} \{ t : \phi\}$ for $\fD$ and $\bar b$ is then defined
as $(\Sigma_{\bar x} \{ t : \alpha\})^\fD[\bar b]:=\sum M_t[\bar b]$. To simplify notation
we omit $\alpha$ in the trivial case that $\alpha=\true$.
The \emph{Presburger fragment} $\fomts[+]$   
of $\fomts[+,\,\cdot\,]$ is defined analogously,  but forbids multiplication
$t\cdot t'$ of numeric terms.

\medskip  To illustrate these logics we first show  how basic multiteam atoms can be expressed in them.
Let $\fD=(\fA,\fN,\{f\})$ be a multiset structure encoding a pair $(\fA,M)$,
with $\dom(M)=\{x_1,\dots,x_k\}$. The size of $M$ 
is expressed by the term $|M|:=\Sigma_{\bar x} \{ f(\bar x) \}$. 
For $i\leq k$, the value of the term $T_{i}(z):=\Sigma_{\bar x} \{ f(\bar x) : x_i=z\}$
on $\fD$, for $z=a$, is  $T_i^\fD[a]=|M_{x_i=a}|$, and terms $T_{i_1\dots,i_\ell}(z_1,\dots, z_\ell)$
for tuples of values $a_1,\dots,a_\ell$ are defined analogously. Then
\begin{itemize}
\item $\fA\models_M x_i\mincl x_j\ \Iff\ \fD\models \A z  \bigl(  T_i(z) =T_j(z) \bigr)$;
\item $\fA\models_M x_i\excl x_j\ \Iff\ \fD\models \A \bar x \A \bar y \big( x_i= y_j \ra ( f\bar x = 0 \lor f\bar y=0)\bigr) $;
\item $\fA\models_M x_i\mindep x_j\ \Iff\ \fD\models \A y\A z \bigl (T_i(y)\cdot T_j(z)= |M|\cdot T_{ij}(y,z) \bigr)$.
\end{itemize}
Notice that while inclusion and exclusion atoms are definable in the Presburger fragment of $\fomts$,
the definition of independence atoms requires multiplication.

To represent multiteam semantics beyond the atomic level, second-order features need to be added.
Different variants for existential second-order logic for metafinite structures (or, equivalently, metafinite spectra) are discussed extensively in \cite{GraedelGur98}.
It makes a difference, whether existential quantification is
used only over relations ranging over the finite structure $\fA$, or, as needed here, over functions $g\colon A^k\ra \N$.
Notice however, that the natural closure of $\fomts[+,\,\cdot\,]$ under existential second-order quantification,
i.e.~the set of sentences $\E g_1\dots\E g_m\phi$, where $\phi$ is in $\fomts[+,\,\cdot\,]$, extended to
weight functions $g_1,\dots,g_m$ rather than just $f$, is a logic in which one can easily 
express all of $\fom[\Omega]$ for any collection $\Omega$ of first-order definable multiset atoms,
but which is  far too powerful, with non-recursive semantics, capturing all r.e.~properties of 
multiteam structures (see \cite[Corollary~4.11]{GraedelGur98}).  Instead we consider a restricted language, 
that limits the range of the quantified functions, and uses only the Presburger fragment of $\fomts$.

\begin{definition}
  We define $\esomts[+]$ as the class of formulae of form
  \[   \psi(f):= (\E g_1\bound b_1)\dots (\E g_m\bound b_m)\phi(f,g_1,\dots,g_m)\]
  where $\phi$ is in the Presburger fragment $\fomts[+]$ and where the bounds $b_i$ are of form $|f|\cdot |A^{\ell_i}|$ for $\ell_i\geq 0$.
  The semantics is the obvious one: $\fD=(\fA,\fN,\{f\})\models \psi(f)$ if there are functions $g_1,\dotsc, g_m$, $g_i\colon A^{r_i}\ra \N$ of appropriate arity $r_i$ such that $|g_i|=\sum_{\ta} g_i(\ta) \leq |f|\cdot |A|^{\ell_i}$ and $(\fD,g_1,\dots,g_m)\models\phi$.
\end{definition}

We first observe that $\esomts[+]$ is powerful enough to capture $\fom[\mincl,{}\excl{}]$.
In fact, we prove a stronger result.

\begin{theorem}
  \label{thm: mts to eso}
  Let $\Omega$ be any collection of multiset atoms that are definable in $\fomts[+]$.
  For every formula $\phi(\bar x)\in\fom[\Omega]$ there exists a sentence $\psi(f)\in\esomts[+]$, such that for every multiteam structure $\fD=(\fA,\fN,\{f\})$ we have that $\fA\models_{M_f} \phi$ if, and only if, $\fD\models \psi(f)$.
\end{theorem}

\begin{proof} If $\phi(\tx)$ is a first-order literal, we translate into $\psi(f):= \A\tx(f\tx = 0\lor\phi(\tx))$.
If $\phi,\phi'\in \fom[\Omega]$ are translated into $\psi(f),\psi'(f)\in \esomts[+]$, respectively,
then
\begin{itemize}
  \item $\fA\models_M \phi\lor\phi'$ if, and only if, $\fD\models (\E g\bound |f|)(\E h\bound |f|)( \A\tx(g\tx+h\tx=f\tx) \land \psi(g)\land\psi'(h))$,
  \item $\fA\models_M \A y \phi$ if, and only if, $\fD\models (\E g \bound |f|\cdot |A|)( \A \tx \A y\ g(\bar x,y)=f(\bar x) \land \psi(g))$,
  \item $\fA\models_M \E y \phi$ if, and only if, $\fD\models (\E g \bound |f|)( \A \tx  (f(\tx)=\Sigma_y \{ g(\bar x,y)\}) \land \psi(g))$.
\end{itemize}
\negline
\end{proof}

It is much more difficult to establish a converse translation.
Since formulae in $\fom$ are always true for the empty multiteam, we only consider
formulae $\psi(f)\in \esomts[+]$ such that $\psi(0)$ is a tautology;
notice that for $f=0$, all existential quantifiers $(\E g_i\bound b_i)$
trivialise because $b_i=0$ and all term equalities reduce to $0=0$.
Thus $\psi(0)$ reduces to a first-order sentence over $\fA$.
If $\psi(0)$ is not a tautology, then we cannot translate $\psi(f)$ into a formula
$\phi\in\fom$ because $f=0$ encodes the empty multiteam and $\fA \models_{\mset{}} \phi$
holds for all $\fA$.

We start by transforming every sentence of $\esomts[+]$ into a normal form which
simplifies the translation of $\esomts[+]$ to $\fom[\mincl, {}\excl{}]$ afterwards.
Further, all structures are assumed to contain at least two elements, hence we may
quantify over Boolean variables (for which we use Greek letters $\mu_i, \chi_i$).

A sentence $\psi(f)\in\esomts[+]$ is in \emph{normal form} if it has the shape
\[
  \psi(f)= (\E g_1\bound b_1)\dots (\E g_m\bound b_m)\A \tx \Lor_{i\leq k} \psi_i
\]
where each $\psi_i$ is a conjunction of first-order literals $\alpha(\bar x)$ and
\emph{simple term equalities}, that are terms of form $g\tx = 0$ or
$\Sigma_{\ty} \{ g\tx\ty \}=\Sigma_{\tv} \{ h\tv\tz \}$.

\begin{lemma}
  \label{lem: eso nf}
  Every sentence $\psi(f)\in\esomts[+]$ such that $\psi(\tuple{0})$ is a tautology
  is equivalent to one in normal form.
\end{lemma}
\begin{proof}
  We assume that $f \neq 0$.
  By renaming variables we can make sure that no re-quantification occurs.
  Towards first-order existential quantification we make use of Skolemisation as follows:
  we substitute a subformula $\E x \psi$ by $(\E f_x\bound |f|)(\Sigma_\tx\{f\tx\} = \Sigma_y\{f_xy\} \land \A x(f_xx=0 \lor \psi))$.
  The function $f_x$ must assign at least one value a non-negative weight because $\Sigma_\tx\{f\tx\} > 0$.

  Let $t(\tx)=t'(\ty)$ be a non simple term occurring in $\phi$, where simple means either $g\tx = 0$ or $\Sigma_{\ty} \{ g\tx\ty \}=\Sigma_{\tv} \{ h\tv\tz \}$.
  We introduce new appropriately bounded function variables $g_{t}$ and $g_{t'}$ (depending on the functions occurring in the terms $t$ and $t'$ and the number of additions) together with the equalities $g_{t}\tx = t(\tx) \land g_{t'}\ty = t'(\ty) \land g_{t}\tx = g_{t'}\ty$ and replace the original equality by this expression.
  Assume $t'$ is either $\Sigma_{\tx}\{t : \alpha\}$ or $g\ty + h\tz$ and occurs in a term other than $f\tv = t'$.
  Introduce the function $g_{t'}$ together with the equality $g_{t'}\ty\tz = t'$ and replace all other occurrences of $t'(\ty\tz)$ by $g_{t'}\ty\tz$.
  For $f'\tx = g\ty + g'\tz$ we again introduce a fresh function variable $h_{g+g'}$ of arity $2$ bounded by $|f|\cdot|A|^{r+r'+1}$, where $r$, resp.~$r'$, refers to the bound of $g$, resp.~$g'$.
  Moreover, we add $\E abcd(ac \neq bd \land h_{g+g'}ac = g\ty \land h_{g+g'}bd=g'\tz \land \A v\A w((vw\neq ac \lor vw\neq bd)\ra h_{g+g'}vw=0))$.
  Hence we can exchange $f'\tx = g\ty + g'\tz$ by $f'\tx = \Sigma_{xy}\{h_{g+g'}xy : \true\}$.
  The last non simple term that still might occur is $f\tx = \Sigma_{\ty}\{g\ty\tz : \alpha\}$.
  The same technique is applied in this case, that is a fresh function variable $h$ is quantified with the same arity and bound as $g$.
  Further the formula $\A\ty\A\tz((\alpha \ra h\ty\tz = g\ty\tz)\land (\neg\alpha \ra h\ty\tz = 0))$ ensures that $h$ behaves as intended, therefore $\Sigma_{\ty}\{g\ty\tz : \alpha\}$ can be exchanged by $\Sigma_{\ty}\{h\ty\tz\}$.
  Lastly, we may of course replace $g\tx$ by $\Sigma_{\ty} \{ g\tx \}$.

  The next task is to bring all second-order quantifiers in front.
  It is clear that function quantification distributes over conjunction and disjunction, but universal quantification has to be treated with care.
  Consider the case $\A x(\E g\bound |f|\cdot|A|^\ell)\psi$.
  Analogously to the translation into Skolem normal form in $\eso$, we increase the arity of $g$ by one (call it $g'$) and its bound by a factor of $|A|$.
  Then we replace all occurrences of $g\ty$ by $g'\ty x$.
  At this point a subtle pitfall appears: since the bound on $g'$ is increased by a factor of $|A|$ in order to store the function $g$ depending on $x$ it might be the case that for some values $a$ of $x$, $\Sigma_\ty g'\ty a > |f|\cdot|A|^\ell$, which is not equivalent to stating that for each $x$ a $g$ of the required bound exists.
  Hence we must make sure that the size of $g'$ for a given $x$ is no larger than that of $g$.
  The resulting formula thus is $(\E g'\bound |f|\cdot|A|^{\ell+1})(\E \hat{g}\bound|f|\cdot|A|^\ell)\A x\,(\Sigma_{\ty}\{g'\ty x\} = \Sigma_{z}\{\hat{g}z\} \land \psi[g\mapsto g'])$.
  Of course, since all variables have unique names we can move all universal quantifiers in front of the (otherwise) quantifier free formula $\psi$ to obtain $(\E g_1\bound b_1)\dots(\E g_k\bound b_k) \A\tx \psi$ and then rewrite $\psi$ to be in disjunctive normal form.

  Finally, we may exchange a bound $|g_i|\cdot|A|^r$ by $|f|\cdot|A|^\ell$ if $g_i$ is bounded by $|f|\cdot|A|^{\ell-r}$ and by iterating this procedure only the required bounds remain.
\end{proof}

\begin{theorem}
  \label{thm:eso to mts}
  For every sentence $\psi(f)\in\esomts[+]$ such that $\psi(0)\equiv\true$,  there exists a formula 
  $\phi\in  \fom[\mincl,{}\excl{}]$ such that, for every multiteam structure $\fD=(\fA,\fN,\{f\})$ where $f$ encodes a non-empty multiteam $M$,
  we have that $\fA\models_M \phi$ if, and only if, $\fD\models \psi(f)$.
\end{theorem}

\begin{proof}
By Lemma \ref{lem: eso nf} we may assume that $\psi$ is equivalent to
$(\E g_1\bound b_1)\dots (\E g_m\bound b_m)\A \tx \Lor_{i\leq k} \psi_i$,
where each $\psi_i$ is a conjunction of first-order literals and simple term equalities.
For a tuple of variables $\tz=z_1,\dots,z_m$ and $r\leq m$ we use the notation
$\tz=(\tz^{\leq r},\tz^{>r})$ where $\tz^{\leq r}=z_1,\dots,z_r$ and $\tz^{>r}=z_{r+1},\dots,z_m$.
To simplify notation we denote $f$ by $g_0$. Let $b_i=|f| \cdot |A|^{\ell_i}$,
and let $\ell=\max(\ell_1,\dots,\ell_k)$.
The desired formula has the form
\[
\phi(\ty_0) \ceq \A z_1\dots \A z_\ell\E\mu_0(\E\ty_1\E\mu_1)\dots (\E\ty_m\E\mu_m)\A\tx\;(\theta\land\psi^\#)
\]
where the $\mu_i$ are Boolean variables and $\ty_i$ is a tuple whose length matches the arity of $g_i$.
A multiteam $M$ over the variables $\ty,\tmu$ encodes the functions $g_0$ to $g_m$
by setting $g_i^M(\ta) \ceq |\mset{s\in M : s(\mu_i)=\true \text{ and } s(\ty_i) = \ta}|$.
The subformula $\theta$ controls the behaviour of the Boolean variables $\mu_i$
that indicate which values for $\ty_i$ are used for representing a function value of $g_i$.
Notice that the multiteam $M[\tz\mapsto A^\ell]$ has size $|f|\cdot|A|^\ell$.
Since this bound may be larger than the bound of $g_i$ we impose in $\theta$ that
only assignments satisfying $\tz^{>\ell_i}=0$ can encode values of $g_i$, by
\[
  \theta \ceq \Land\nolimits_{i\leq m} \bigl(\bigl( \tz^{>\ell_i}=\bar 0 \imp
  \dep(\tz^{\leq \ell_i}, \mu_i)\bigr) \land (\mu_i \ra  \tz^{>\ell_i}=\bar 0 )\bigr).
\]
Let us describe $\psi^\#$.
By assumption the quantifier-free part of $\psi$ is a disjunction over formulae
$\psi_i$ each of which is a  conjunction of  literals and simple term equalities.
Thus, we have to split all values for the variables $\tx$ such that each one
satisfies (at least one) formula $\psi_i$.
To simulate this, we quantify over Boolean variables $\chi_i$ with the indented
semantics that $\chi_i$ holds in an assignment $s$ if $s(\tx)$ is distributed to $\psi_i$.
Recall that $\alpha \imp \beta$ is a shorthand for $\neg\alpha \lor (\alpha\land\beta)$.
Formally,
\[
  \psi^\# \ceq \E \chi_1\dots\E\chi_k (\Land_{i\leq k} \dep(\tx, \chi_i) \land
  \Lor_{i\leq k}\chi_i \land \Land_{i\leq k}(\chi_i \ra \psi^\#_i)),
\]
where $\psi_i^\#$ is the conjunction of the literals in $\psi_i$ (without changes) and appropriate translations of its term equations:
\begin{itemize}
  \item Equations of form $g_i\tu = 0$ are translated into $\mu_i\ra\tu\neq\ty_i$.
  \item An equation $e(\tx)$ of form  $\Sigma_{\tv} \{ g_i\tu\tv \}=\Sigma_{\tw} \{ g_j\tu'\tw \}$ where $\tu,\tu'$ are tuples of length $r$ and $s$, respectively, of (not necessarily distinct) variables from $\tx$ is translated into
  \[
  \E \nu\E \lambda \bigl(
  (\nu \lra (\mu_i\land  \tu=\ty_i^{\leq r}))\land 
  (\lambda \lra (\mu_j\land  \tu'= \ty_j^{\leq s})) \land \tx \nu\mincl\tx \lambda \bigr).
  \]
\end{itemize}
Notice that the variables $\tx$ are universally quantified and hence the functions $g_i$ are represented multiple times
in the multiteam on which $\psi^\#$ is evaluated. But this is not a problem because the inclusion statement for $\nu$ and $\lambda$
must hold for every value of $\tx$ (cf.~Proposition \ref{prop: mincl basic properties}).

It remains to prove that for the translation of a
sentence $\psi(f)\in\esomts$ in normal form into $\phi\in\fom[\mincl,{}\excl{}]$, we indeed
have that $\fA\models_M \phi$ if, and only if, $\fD\models \psi(f)$,
for every $\fD=(\fA,\fN,\{f\})$ where $f$ encodes a non-empty multiteam $M$.
Recall that 
\begin{align*}
  \psi(f) &= (\E g_1\bound b_1)\dots (\E g_m\bound b_m)\A \tx \Lor_{i\leq k} \psi_i &\text{and}\\
  \phi(\ty_0) &= \A z_1\dots \A z_\ell \E \mu_0(\E\ty_1\E\mu_1)\dots (\E\ty_m\E\mu_m)\A\tx\;(\theta\land\psi^\#), &\text{with}\\
  \theta &= \Land_{i\leq k} \bigl(\bigl( \tz^{>\ell_i}=\bar 0 \imp \dep(\tz^{\leq \ell_i}, \mu_i)\bigr) \land
    \bigl(\mu_i \imp \tz^{>\ell_i}=\bar 0 \bigr)\bigr) &\text{and}\\
  \psi^\# &= \E \chi_1\dots\E\chi_k (\Land_{i\leq k} \dep(\tx, \chi_i) \land \Lor_{i\leq k}\chi_i \land
    \Land_{i\leq k}(\chi_i \imp \psi^\#_i)).
\end{align*}

Assume that $\fD\models\psi(f)$ that is $(\fD, g_1,\dotsc,g_m) \models \A\tx \bigvee_{i\leq k}\psi_i$ 
for appropriate functions $g_i$.
The evaluation of $\phi$ on $M$ leads to an expanded multiteam
$M_1=M[\tz\mapsto A^\ell,\mu_0\mapsto F_0,(\ty_1,\mu_1)\mapsto F_1,\dots,(\ty_m,\mu_m)\mapsto F_m]$
for appropriate choice functions $F_0,\dots,F_m$.
Recall that $M_1$ encodes $g_i$ if, for all $\ta$, we have that $g_i(\ta)=
|\mset{s\in M_1 : s(\mu_i)=\true \text{ and } s(\ty_i) = \ta}|$.
It is clear that choice functions $F_0,\dots,F_m$ exist, such that $M_1$ satisfies
$\theta$ and  indeed encodes $g_0$ up to $g_m$ in this sense.
Notice however, that $M_1$ is then further expanded by the universal quantifiers
over $\tx$ to $M_2=M_1[\tx\mapsto A^r]$ which then encodes the functions $|A^r|\cdot g_i$ rather than $g_i$.
However, for each value $\tc\in A^r$ the restricted multiteam
$M_1[\tx\mapsto A^r]_{\tx=\tc}$ again encodes $g_0,\dots,g_m$.

Let $A^r = U_1\cupdot\dotsb\cupdot U_{k}$ such that
$(\fD,g_1,\dots,g_k) \models  \psi_i(\tu)$ for all $\tu\in U_i$.
We claim that $\fA \models_{M_2}\psi^\# = \E\chi_1\dots\E\chi_{k}\tilde{\psi}$.
Consider the choice function $F$ such that for each $s \in M_3 \ceq M_2[\chi_1,\dotsc,\chi_{k}\mapsto F]$ we have
$s(\chi_i)=\true$ if, and only if, $s(\tx)\in U_i$.
Clearly, the dependencies $\dep(\tx, \chi_i)$ and the disjunction over all $\chi_i$ are satisfied in $M_3$.
It remains to show that indeed $\fA \models_{M_3}  \chi_i\ra\psi^\#_i$, which means that
$\fA \models_{M(i)} \psi^\#_i$ where $M(i):=\mset{s\in M_3:  s(\chi_i)=\true}$.
Notice that $s(\tx)\in U_i$ for all $s\in M(i)$.
Hence, for any first-order literal $\alpha(\tx)$ in $\psi^\#_i$, the same literal is also in
$\psi_i$ and hence $\fA\models \alpha[s]$ for $s\in M(i)$.
Recall that equations $e(\tx)$ in $\psi_i$ of form $g_i\tu = 0$ (for $\tu\subseteq\tx$)
are translated into $\mu_i\ra\tu\neq\ty_i$.
Since $(\fD,g_i)\models e(\ta)$ for all $\ta\in U_i$, and $M_1$ encodes $g_i$
there cannot be any $s\in M(i)$ with $s(\mu_i)=\true$ and $s(\tu)=s(\ty_i)$.
Finally consider equations $e(\tx)$ in $\psi_i$ of form $\Sigma_{\tv} \{ g_k\tu\tv \}=\Sigma_{\tw} \{ g_j\tu'\tw \}$
and their translations  in $\phi$ by
\[  e^\#:=\E \nu\E \lambda \bigl(
      (\nu \lra (\mu_k\land  \tu=\ty_k^{\leq r}))\land 
       (\lambda \lra (\mu_j\land  \tu'= \ty_j^{\leq s})) \land \tx \nu\mincl\tx \lambda \bigr).\]
By assumption the functions $g_k$ and $g_j$  satisfy $e(\ta)$ for $\ta\in U_i$.
To see that $e^\#$ is true in $M(i)$, we expand
all assignments $s\in M(i)$ with $s(\mu_k)=\true$ and $s(\tu)=s(\ty_k^{\leq r})$ to $s[\nu\mapsto \true]$,
and all others to $s[\nu\mapsto \false]$, and analogously for $\lambda$.
This produces a new multiteam $\tilde{M}(i)$, and  we have to prove that it
satisfies the inclusion  $\tx \nu\mincl\tx \lambda$.
Consider any value $\ta\in U_i$ for $\tx$; it induces values $\tc$ and $\tc'$ for the subtuples $\tu$ and $\tu'$ of $\tx$.
We have that
\begin{align*}
  & |\mset{s \in \tilde{M}(i) : s(\mu_k)=\true \text{ and }s(\tu)=\tc=s(\ty_k^{\leq r})}| = \Sigma_{\tv}\{g_k\tc\tv\}, \text{ and }\\
  & |\mset{s \in \tilde{M}(i) : s(\mu_j)=\true \text{ and }s(\tu')=\tc'=s(\ty_j^{\leq s})}| = \Sigma_{\tw}\{g_j\tc'\tw\}.
\end{align*}
Since the two values are same it follows that, for any tuple $\ta\in U_i$ we have that $\tilde{M}(i)_{\tx=\ta}$
satisfies the inclusion $\nu\mincl\lambda$.
Hence, by Proposition \ref{prop: mincl basic properties}, $\tilde{M}(i)$ satisfies $\tx \nu\mincl\tx \lambda$.

For the converse, assume that we have a structurally valid annotation of the syntax tree $\cT(\phi)$ 
witnessing that $\fA\models_M \phi$.
It assigns to every subformula $\eta$ of $\phi$ a multiteam $M(\eta)$ with $\fA\models_{M(\eta)} \eta$.
Thus $M(\A\tx(\theta \land \psi^\#))$ encodes functions $g_0$ to $g_k$ satisfying the
bounds $b_i$.  We claim that $(\fD, g_1,\dots,g_k) \models \A\tx \bigvee_{i\leq k}\psi_i$.
As above the values quantified for the $\chi_i$ give rise to a (a fortiori) disjoint cover $A^k = U_1\cupdot\dotsb\cupdot U_{k}$ where $\ta\in U_i$
whenever there is an assignment $s\in M(\psi^\#)$ with $s(\chi_i)=\true$ and $s(\tx)=\ta$.
Notice that $M(\psi_i^\#)=\mset{s\in M(\psi^\#):  s(\chi_i)=\true}$ satisfies $\psi_i^\#$.
We claim that $(\fD,g_1,\dots,g_m) \models \psi_i(\ta)$ for $\ta\in U_i$.
For first-order literals $\alpha(\tx)$ or equalities $g_j\tu=0$ in $\psi_i$, this follows directly from the fact that
 $\fA\models_{M(\psi_i^\#)} \psi_i^\#$.
 Finally, we consider equations $e(\tx)$  in $\psi_i$ of   form  $\Sigma_{\tv} \{ g_k\tu\tv \}=\Sigma_{\tw} \{ g_j\tu'\tw \}$.
 By assumption, the multiteam $M(\psi_i^\#)=M(e^\#)$ satisfies $e^\#$; hence there is a choice function $G$
 such that $M(\tx \nu\mincl\tx \lambda)=M(e^\#)[\nu\lambda\mapsto G]$ satisfies the inclusion $\tx \nu\mincl\tx \lambda$,
 which implies that its restrictions  $M(\tx \nu\mincl\tx \lambda)_{\tx=\ta}$ satisfy $\nu\mincl\lambda$.
Hence, on any such restriction, the choice function $G$ maps $\nu$ and $\lambda$ to $\true$ for an equal number of assignments.
By the stipulation on these values these numbers correspond to $\Sigma_{\tv} \{ g_k\tu\tv \}$ and $\Sigma_{\tw} \{ g_j\tu'\tw \}$, respectively.
Therefore $(\fD,g_1,\dots,g_m) \models e(\ta)$.
\end{proof}

By Theorem~\ref{thm: mts to eso} and \ref{thm:eso to mts}, we have a kind of equivalence between the multiteam logic $\fom[\,\excl\,,\mincl]$ and the existential second-order logic $\esomts[+]$.
Since we know that independence cannot be expressed in $\fom[\mincl,{}\excl{}]$, we infer that it also cannot be defined in $\esomts[+]$.
As we have seen, there is a simple translation of independence atoms into the stronger
variant of $\esomts$ that uses multiplication. However, this stronger logic cannot be translated back
into a logic with multiteam semantics. Indeed, with multiplication at hand it is not difficult to
write a formula stating that $|f| = 1$. But this implies that the logic is not closed under
multiplication of (functions that represent) multiteams.
To come up with an extension of $\esomts[+]$ that would capture $\fom[\mindep]$ we would have to
add a controlled form of multiplication that avoids such problems. It is not difficult to define
an operator that is essentially a direct translation of independence into the setting of
multiteam structures, but it is unclear at this point whether this is possible with a \emph{natural} algebraic operator.

\section{Multiteams with real weights}
\label{sec: real}

We now turn to a different notion of multiteams where the assignments do not come
with a multiplicity in $\N$, but with a \emph{weight}. Although our weight functions need not be probability 
distributions and can take arbitrary values in $\R_{>0}$,
this is related to the notion of a probabilistic team as in \cite{DurandHanKonMeiVir18a, HannulaHirKonKulVir19},
as long as we only consider dependencies that are invariant under
multiplication. This has also been noted in \cite{HannulaHirKonKulVir19}.

\begin{definition}
  A \emph{real valued multiteam} is a pair $(X, r)$ where $X$ is a team and $r\colon X\to\R_{>0}$.
\end{definition}

Notions such as $|M|, M \uplus N, M\subsetpeq N$ and so on are defined as for natural multiteams.
We view a splitting $(X, r) = R_1 \uplus R_2$ as the action of a function
$\cS\colon X \to [0, 1]$, with $R_1 = \cS\cdot (X, r) \ceq (\{s \in X : \cS(s) > 0\},
s\mapsto r(s) \cdot\cS(s))$ and $R_2 = (1-\cS)\cdot(X, r)$.
Regarding multiteam extensions we put $M[x \mapsto A] = (X[x\mapsto A], r'\colon s[x\mapsto a] \mapsto r(s))$
and $M[x\mapsto F] = (\{s[x\mapsto a] : s\in X, F(s)(a) > 0\}, r')$ with
$r'\colon s[x\mapsto a] \mapsto r(s) \cdot F(s)(a)$, where $F\colon X \to [0, 1]^A$
is a choice function, that is for all $s\in X$, $\sum_{a\in A}F(s)(a) = 1$ holds.

The semantics of  logical connectives for real valued multiteams is analogous to the one given in Definition~\ref{def: multiteam semantics}.
It only differs slightly for disjunction and existential quantification:
\begin{itemize}
  \item $\fA\models_M \phi_1\lor\phi_2 \defiff \fA\models_{\cS\cdot M}\phi_1$ and $\fA\models_{(1-\cS)\cdot M}\phi_2$ for some $\cS\colon X\mapsto [0, 1]$;
  \item $\fA\models_M \E x \phi \defiff \fA\models_{M[x\mapsto F]} \phi$ for some choice function $F\colon X\to[0, 1]^A$.
\end{itemize}
Next we define a notion of approximation for multiteams and the concepts of
topologically open and closed sets of multiteams and formulae.
The intention is that an open formula $\phi$ has the property that if $M$ satisfies
$\phi$ and $N$ is sufficiently close to $M$ then also $N$ should satisfy $\phi$.
A topologically closed formula $\psi$ will have the property that whenever all
elements of a sequence of multiteams satisfy $\psi$, then also their limit does.

\begin{definition}
  A multiteam $N = (Y, t)$ approximates $M = (X, r)$ with error $\error{M}(N)$
  if $Y \subseteq X$ and $\error{M}(N)=\sum_{s \in X} |r(s) - t(s)|$.
  If $Y\nsubseteq X$, we put $\error{M}(N)=\infty$.
\end{definition}

Notice that we only approximate multiteams ``from below'', i.e., whenever $\error{(X, r)}(Y, t) < \infty$ then $Y \subseteq X$.
The reason for this choice is that first-order formulae are preserved by subteams but not by superteams.
Hence, if a multiteam $N$ has more assignments (in its support) than $M$ then the truth of $\phi$ in $M$
gives us no indication whether $\phi$ holds in $N$.
The same holds for the opposite notion.
Let $(M_i)_{i\in\N}$ be a sequence of multiteams $M_i = (X, r_i)$.
If for every $s\in X$ the sequence $(r_i(s))_{i\in\N}$ converges we define $\lim\limits_{i\to\infty}M_i \ceq (Y, r)$, where $r(s) \ceq \lim\limits_{i\to\infty} r_i(s)$ and $Y = \{s\in X : r(s) > 0\}$.
We call $(M_i)_{i\in\N}$ convergent if its limit exists.

\begin{definition}
  \label{def: clopen}
  A set  $\cM$ of multiteams is called
  \begin{itemize}
    \item \emph{open} if for every $M \in \cM$ there is some $\epsilon > 0$ such that $\error{M}(N) < \epsilon$ implies $N \in \cM$, and
    \item \emph{closed} if for every convergent sequence $(M_i)_{i\in \N}$ of multiteams $M_i \in \cM$ also $\lim\limits_{i\to\infty} M_i \in \cM$.
  \end{itemize}
\end{definition}

We call a formula $\psi$ \emph{topologically open} or \emph{closed} if for every structure $\fA$ the set
$\psi^\fA = \{M : M \text{ is a multiteam over } A \text{ with } \free(M) = \free(\psi) \text{ and } \fA\models_M\psi\}$ is open, resp.\ closed.
We write $\overline{\psi^\fA}$ for the analogous set with $\fA\nmodels_M\psi$.
As usual, a formula is called \emph{clopen}, if it is both open and closed.
Open and closed sets are related in the usual way, if we tighten the definitions to allow only multiteams of the same support (see the discussion above).

\begin{lemma}
  Let $\fA$ be a structure, $\psi$ be a formula and $X$ a team.
  Then $\psi^\fA\res{X} = \{M \in \psi^\fA : \supp{M} = X\}$ is open if $\overline{\psi^\fA}\res{X} = \{M \in \overline{\psi^\fA} : \supp{M} = X\}$ is closed and vise versa.
\end{lemma}
\begin{proof}
  Let $\psi^\fA\res{X}$ be open and $(M_i)_{i\in\N}$ a convergent sequence of multiteams in $\overline{\psi^\fA}\res{X}$.
  For the sake of the argument assume $M\ceq\lim\nolimits_{i\to\infty}M_i \notin \overline{\psi^\fA}\res{X}$.
  Then $M \in \psi^\fA\res{X}$, thus there is some $\epsilon > 0$ such that all multiteams with approximation error smaller than $\epsilon$ to $M$ also satisfy $\psi$ and hence belong to $\psi^\fA\res{X}$.
  But this means that $\lim\nolimits_{i\to\infty} M_i \neq M$.
  
  On the other hand, if $\psi^\fA\res{X}$ is closed, assume that $M \in \overline{\psi^\fA}\res{X}$ but for all $\epsilon > 0$ there is some $N$ with $\error{M}(N) < \epsilon$ and $\fA\models_N \psi$.
  Hence there is a sequence of multiteams $(N_i)_{i\in\N}$ approaching $M$ satisfying for all $i\in\N$, $\fA\models_{N_i} \psi$.
  Therefore $\lim\nolimits_{i\to\infty} N_i = M \in \psi^\fA\res{X}$.
\end{proof}

The logical connectives preserve topologically open and closed formulae.

\begin{lemma}
  \label{lem: preserve clopen}
  Let $\phi$ and $\psi$ be topologically open, resp.\ closed, formulae.
  Then $\phi\land\psi$, $\phi\lor\psi$ and $\A x\phi$ and $\E x\phi$ are
  topologically open, resp.\ closed, as well.
\end{lemma}
\begin{proof}
  The proof is carried out by structural induction on the formulae.
  For simplicity we assume that no variable is requantified and that the structure
  $\fA$ is fixed.
  First, we proof that topologically open formulae are preserved.
  Let $\fA\models_M\phi_i$ and $\epsilon^M_{\phi_i} > 0$ depend on $\fA$, $M$ and
  $\phi_i$ such that all multiteams $N$ that approximate $M$ up to an error smaller
  than $\epsilon^M_{\phi_i}$ also satisfy $\phi_i$ under $\fA$.
  
  For conjunction and disjunction we put $\epsilon^M_{\phi_1\circ\phi_2} \ceq
  \min(\epsilon^M_{\phi_1}, \epsilon^M_{\phi_2})$.
  By induction hypothesis it is clear that this bound suffices for the conjunction.
  If $\fA\models_M\phi_1\lor\phi_2$ then there is a splitting $\cS\colon X \to [0, 1]$
  for $M$ s.t.\ $\fA\models_{M_i}\phi_i$ with $M_i = \cS\cdot M$ (resp.\ $1-\cS$).
  Let $N$ be a multiteam with $\error{M}(N) < \epsilon^M_{\phi_1\lor\phi_2}$.
  We claim that $\cS$ is a splitting for $N$ such that $\fA\models_{N_i}\phi_i$
  for $N_1 = \cS\cdot N$ and $N_2 = (1-\cS) \cdot N$.
  Indeed, notice that $\error{\cS\cdot M}(\cS\cdot N) \leq \error{M}(N)$, hence
  $\fA\models_{N_i}\phi_i$ for $i=1,2$.
  
  Assume $\fA\models_M\A x\phi$.
  Thus, for $M' = M[x\mapsto A]$ we have $\fA\models_{M'}\phi$ and hence by induction
  hypothesis a bound $\epsilon_{\phi}^{M'} > 0$ with the required property exists.
  Put $\epsilon_{\A x\phi}^{M} \ceq \epsilon_{\phi}^{M'}$ and let $N$ be a multiteam
  in the $\epsilon_{\phi}^{M'}$ environment of $M$, that is $\error{M}(N) < \epsilon_{\phi}^{M'}$.
  Further, since for $N' = N[x\mapsto A]$ the errors of $N$ to $M$ and $N'$ to $M'$
  are identical, that is $\error{M'}(N') = \error{M}(N)$, and by $\fA\models_{N'}\phi$,
  which follows from $\phi$ being topologically open, we can conclude that also
  $\fA\models_N\A\phi$ as claimed.
  
  Finally, assume $\fA\models_M\E x \phi$.
  That means a choice function $F$ exists such that $\fA\models_{M'}\phi$, where
  $M' = M[x\mapsto F]$.
  As in the previous case let $\epsilon_{\E x\phi}^M \ceq \epsilon_\phi^{M'}$ and
  $N$ be such that $\error{M}(N) < \epsilon_{\E x\phi}^M$.
  For the restriction $F\res{N}$ of $F$ to $N$ we obtain $\error{M}(N) = \error{M'}(N[x\mapsto F\res{N}])$
  which, by induction hypothesis, implies $\fA\models_{N[x\mapsto F\res{N}]}\phi$.
  
  Before we turn our attention to topologically closed formulae and show that the
  logical connectives preserve also this property we recall the Bolzano-Weierstra\ss\ Theorem:
  Every bounded sequence in $\R^n$ has a convergent subsequence.
  
  Let $\psi$ be a topologically closed formula, $(M_i)_{i\in\N}$ be a convergent
  multiteam sequence over the support $X$ such that $\fA\models_{M_i}\psi$ holds
  for all $i\in\N$ and $M = \lim\limits_{i\to\infty}M_i$.
  \begin{itemize}
    \item Regarding conjunction, the property follows directly from induction hypothesis.
    \item Consider $\psi = \psi_1\lor\psi_2$.
    For every $i\in\N$ let $\cS_i \colon X \to [0, 1]$ be a splitting function
    witnessing $\fA\models_{M_i}\psi$.
    By the Bolzano-Weierstra\ss\ Theorem there is a convergent subsequence
    $(\cS_{n_i})_{i\in\N}$ of $(\cS_i)_{i\in\N}$.
    Hence we may apply the induction hypothesis on $(\cS_{n_i}\cdot M_{n_i})_{i\in\N}$
    and $((1-\cS_{n_i})\cdot M_{n_i})_{i\in\N}$ to conclude that $\fA\models_{M}\psi_1\lor\psi_2$.
    \item For a universally quantified formula $\forall x\phi$ the claimed property
    follows directly from an application of the induction hypothesis.
    \item Let $\exists x \phi$ be given where $\phi$ is a topologically closed formula.
    Thus $(F_i)_{i\in\N}$ exists such that for all $i\in\N$ the structure $\fA$
    satisfies $\phi$ under $M_i[x\mapsto F_i]$, in symbols $\fA\models_{M_i[x\mapsto F_i]}\phi$.
    Since $A$ is finite and $F\colon X \to [0, 1]^A$ we may apply the
    Bolzano-Weierstra\ss\ Theorem on the sequence $(F_i)_{i\in\N}$ and obtain a subsequence
    $(F_{n_i})_{i\in\N}$ s.t.\ $\lim\limits_{i\to\infty}(F_{n_i}) = \tilde{F}\colon Y\to[0,1]^A$ exists.
    The induction hypothesis is now applicable to the sequence $(M_{n_i}[x\mapsto F_{n_i}])_{i\in \N}$.
  \end{itemize}
\end{proof}

\begin{example}
  For all (pairwise distinct) variable tuples $\tx, \ty, \tz$ the following atoms are
  \begin{description}
    \item[open but not closed:] $\tx\fork{<q}\ty$ and $\tx\fork{>q}\ty$ for $q\in (0, 1)$,
    \item[closed but not open:] $\tx\mincl\ty$, $\tx\mindep_{\tz}\ty$ and $\tx\fork{\leq q}\ty$ and $\tx\fork{\geq q}\ty$ for $q\in (0, 1)$,
    \item[clopen:] $\dep(\tx, y), \tx\excl\ty$ and all first-order literals.
  \end{description}
  \negline
\end{example}
\begin{proof}
  The simple case are the downwards closed atoms $\delta$ (which includes first-order
  formulae and hence literals).
  It suffices to notice that if $\error{M}(N)$ is finite then $\supp{N} \subseteq \supp{M}$
  which implies that downwards closed atoms are topologically open.
  Furthermore, if all multiteams $M_i$ of a convergent sequence $(M_i)_{i\in\N}$
  satisfy $\delta$, then the for the limit $M = \lim\nolimits_{i\to\infty}M_i$
  holds $\supp{M} \subseteq \supp{M_i}$ for all $i\in\N$ which again implies the
  claim that $\delta$ is topologically closed.
  
  For the atoms claimed to be open, but not closed, take $\tx\fork{<p}\ty$ as an
  example and let $M$ be a multiteam satisfying this atom in $\fA$.
  Among all tuples $\ta\in\supp{M}(\tx)$ and $\tb\in\supp{M}(\ty)$ let $\ta_m$
  and $\tb_m$ be such that $\pr_M(\ty=\tb\mid\tx=\ta)$ is maximal and put
  $A=|M_{\tx\ty=\ta_m\tb_m}|$, $B=|M_{\tx=\ta_m}|$ and $m = pB-A$ (i.e.~$\pr_M(\ty=\tb\mid\ta=\tx) = A/B$).
  Pick an arbitrary $\epsilon$ satisfying the range condition $0<\epsilon<\frac{m}{1+p}$.
  Let $N$ be a real weight multiteam with $\error{M}(N) < \epsilon$ and $\ta\in\supp{N}(\tx)$ and $\tb\in\supp{N}(\ty)$.
  Notice that $\pr_N(\ty=\tb\mid\tx=\ta)$ is bounded from above by $\frac{A+\epsilon}{B-\epsilon}$.
  A basic calculation shows that $\pr_N(\ty=\tb\mid\tx=\ta) < p$:
  Indeed, $\frac{A+\epsilon}{B-\epsilon} < p \iff \epsilon < pB - A - p\epsilon = m-p\epsilon \iff 1 < \frac{m}{\epsilon} - p \iff \epsilon < \frac{m}{1+p}$, which by definition of $\epsilon$ holds.
  
  Regarding the closed atoms take $\tx\mincl\ty$ as an example.
  A multiteam $(X_i,r_i)$ satisfies this atom if for all appropriate $\ta$ the
  equation $\sum_{s\in X_i, s(\tx)=\ta}r_i(s) = \sum_{s\in X_i, s(\ty)=\tb}r_i(s)$ holds.
  Further, consider any convergent sequence of multiteams $(M_i)_{i\in\N}$ that each satisfy $\tx\mincl\ty$.
  It follows directly that the equation above must be satisfied by the limit of this sequence as well (as we can take the limit on both sides of the equation).
  The claim holds for the other atoms as well, since independence is similarly defined over an equation and the forking atoms are satisfied also for equality.
\end{proof}

It follows that logics based on topologically open, resp.\ closed, atoms are expressively separable.
Towards the question concerning which properties of the supporting team can be defined in a logic with real multiteam semantics
we explore the properties of clopen formulae.
Let $\Omega$ be a set of (multiplicity invariant) atomic dependency notions.

\begin{prop}
  \label{prop: clopen}
  For any $\psi \in \fo[\Omega]$ the following are equivalent.
  \begin{enumerate}
    \item\label{prop: clopen - a} $\psi$ is clopen.
    \item\label{prop: clopen - b} For all structures $\fA$ and all $M, N$ with $\error{M}(N) < \infty$ we have that
    $\fA\models_M \psi$ implies $\fA\models_N \psi$.
    \item\label{prop: clopen - c} $\psi$ is downwards closed.
  \end{enumerate}
\end{prop}
\begin{proof}
  The implications ``\ref{prop: clopen - b} $\implies$ \ref{prop: clopen - a}'' and ``\ref{prop: clopen - b} $\implies$ \ref{prop: clopen - c}'' are immediate.
  Regarding ``\ref{prop: clopen - c} $\implies$ \ref{prop: clopen - b}'' assume $\fA\models_M \psi$,
  and let $N$ be a multiteam with $\error{M}(N) < \infty$, that means $\supp{N} \subseteq \supp{M}$.
  This implies $\fA\models_{|N|\cdot M}\psi$ and by downwards closure of $\psi$ also $\fA\models_N\psi$ as required.
  Towards ``\ref{prop: clopen - a} $\implies$ \ref{prop: clopen - b}'' assume that $\fA\models_M\psi$.
  For a fixed support $X$ there is a bijection $\rho$ mapping every multiteam with support $Y\subseteq X$ to an element of $\R_{\geq0}^n$ where $n=|X|$.
  Indeed, for a fixed ordering of $X = \{s_1,\dots,s_n\}$, every multiteam $(Y, r)$ can be written as a tuple $t \in \R_{\geq0}^n$, where the $i$th entry of $t$ is $r(s_i)$.
  We associate the set $\cT = \{t \in \R_{\geq0}^n : \fA\models_{\rho^{-1}(t)}\psi\}$ with $\fA$, $\psi$ and $X$.
  Then, $\cT$ is clopen in the classical sense if, and only if, $\psi^\fA\res{X}$ is clopen.
  It is well known that the only clopen sets of a connected topological space are $\emptyset$ and the whole space itself.
  Because $\fA\models_M\psi$ it follows that $\cT\neq\emptyset$, whence $\cT = \R_{\geq0}^n$.
\end{proof}

\begin{cor}
  \label{cor: real ts -> downwards closed}
  Every topologically closed formula $\psi\in\fom[\Omega]$ defining a team semantical property, meaning that $\fA\models_{(X, r)}\psi$ if, and only if, $\fA\models_{(X, r')}\psi$ holds for all $\fA, X, r$ and $r'$, is downwards closed.
\end{cor}
\begin{proof}
  The assumption on $\psi$ implies that $\psi$ must be topologically open as well, hence by Proposition \ref{prop: clopen} it follows at once that $\psi$ is downwards closed.
\end{proof}

This corollary does not hold for open formulae.
Indeed, consider $x\fork{<1}y$ which expresses that every value of $x$ in the multiteam 
appears with at least two values for $y$ in common assignments.
This property clearly depends only on the team and holds regardless of the involved multiplicities (because they are all greater than 0).
In logics with team semantics this atom is known as \emph{anonymity} $x\anon y$, which is union closed but not downwards closed.

As an application of this theory we show that the converse of Proposition \ref{prop: psi dc => psi^T and psi companions} holds.
To avoid confusion we write $\models^\mathsf{T}$, $\models^\N$, and $\models^\R$, for team semantic, multiteam semantics under natural and real numbers, respectively.
Let $\Omega$ be a collection of topologically closed atoms.
\begin{prop}
  \label{prop: psi^T and psi companions <=> psi dc}
  A formula $\psi \in \fom[\Omega]$ is downwards closed if, and only if, $\teamVer{\psi}$ and $\psi$ are companions. 
\end{prop}
\begin{proof}
  It remains to prove the direction from right to left.
  Assume that $\fA\models^\N_{(X, n)} \psi$.
  Then, a fortiori, $\fA\models^\R_{(X, n)}\psi$.
  Since this holds for arbitrary $n\colon X\to\N_{>0}$ we may apply Corollary \ref{cor: real ts -> downwards closed} (from $\fA\models^\N_{(X, n)} \psi$ for all $n\colon X\to\N_{>0}$ follows $\fA\models^\R_{(X, r)}\psi$ for all $r\colon X\to\R_{>0}$ because $\psi$ is topologically closed).
  Hence for all $Y\subseteq X$, $\fA\models^\R_{(Y, t)}\psi$ and, in general, $\fA\models^\R_{(X, r)} \theta$ implies that $\fA\models^\mathsf{T}_X \teamVer{\theta}$.
  Thus, $\fA\models^\mathsf{T}_Y \teamVer{\psi}$ and by the fact that $\teamVer{\psi}$ and  $\psi$ are companions we conclude that $\fA\models^\N_{(Y, n)}\psi$.
\end{proof}

\bibliographystyle{plain}
\bibliography{arxiv}

\end{document}